\documentclass[11pt]{article}

\usepackage[T1]{fontenc}
\usepackage[utf8]{inputenc}

\usepackage[tt=false,type1=true]{libertine}
\usepackage[varqu]{zi4}
\usepackage{amsthm,mathtools}
\usepackage[libertine]{newtxmath}
\usepackage{microtype}

\usepackage[
  a4paper,
  margin=1in,
  headheight=14pt,
  headsep=24pt,
  footskip=30pt
]{geometry}

\usepackage{fancyhdr}
\usepackage{amsmath}
\usepackage{booktabs}
\usepackage{graphicx}
\usepackage{url}
\usepackage[hypertexnames=false]{hyperref}
\usepackage{xcolor}
\usepackage{lscape}
\usepackage{float}
\usepackage{makecell}
\usepackage{seqsplit}
\usepackage{array}

\newcolumntype{L}[1]{>{\raggedright\arraybackslash}p{#1}}

\usepackage{tabularx}
\usepackage{adjustbox}
\usepackage[section]{placeins}
\usepackage{enumitem}
\usepackage{pgfplots}
\usepackage{siunitx}
\usepackage[nameinlink,noabbrev]{cleveref}

\pgfplotsset{compat=1.18}
\newtheorem{theorem}{Theorem}[section]
\newtheorem{corollary}[theorem]{Corollary}
\newtheorem{lemma}[theorem]{Lemma}
\newtheorem{proposition}[theorem]{Proposition}
\newtheorem{definition}[theorem]{Definition}
\newtheorem{remark}[theorem]{Remark}
\newtheorem{observation}[theorem]{Observation}

\newcommand{\KeyGen}{\mathsf{KeyGen}}
\newcommand{\Encaps}{\mathsf{Encaps}}
\newcommand{\Decaps}{\mathsf{Decaps}}
\newcommand{\EncC}{\mathsf{EncC}}
\newcommand{\Dec}{\mathsf{Dec}}
\newcommand{\Pass}{\mathsf{Pass}}
\newcommand{\Det}{\mathsf{Det}}

\newcommand{\AvgHinf}{\widetilde{H}_{\infty}}

\newcommand{\Trace}{\mathsf{Trace}}

\newcommand{\VSFO}{\mathsf{VSFO}}

\newcommand{\Hit}{\mathsf{Hit}}

\newcommand{\Inst}{\mathsf{Inst}}
\newcommand{\HRef}{\mathsf{HRef}}
\newcommand{\HSelf}{\mathsf{HSelf}}
\newcommand{\LcUP}{\mathsf{L\mbox{-}cUP}}
\newcommand{\Adv}{\mathsf{Adv}}
\newcommand{\Corr}{\mathsf{Corr}}
\newcommand{\Acc}{\mathsf{Acc}}
\newcommand{\Alias}{\mathsf{Alias}}
\newcommand{\Can}{\mathsf{Can}}
\newcommand{\Good}{\mathsf{Good}}
\newcommand{\Bad}{\mathsf{Bad}}
\newcommand{\TailEnt}{\ensuremath{\mathsf{HQC\mbox{-}TailEnt}}}
\newcommand{\Hcd}{H_{\mathsf{cd}}}
\newcommand{\Obs}{\mathsf{Obs}}
\newcommand{\Tgt}{\mathsf{Tgt}}
\newcommand{\Cone}{\mathsf{Cone}}

\newcommand{\Exec}{\mathsf{Exec}}
\newcommand{\Cert}{\mathsf{Cert}}
\newcommand{\TV}{\mathsf{TV}}
\newcommand{\ctsel}{\mathsf{ct\_select}}
\newcommand{\supp}{\mathsf{supp}}
\newcommand{\bits}{\{0,1\}}

\newcommand{\Wit}{W}

\newcommand{\Dhonref}{D_{\mathrm{honest\_ref}}}
\newcommand{\Dhonself}{D_{\mathrm{honest\_self}}}
\newcommand{\Dctmal}{D_{\mathrm{ctmal}}}
\newcommand{\Ddiff}{D_{\mathrm{diff}}}
\newcommand{\Dkeyfmt}{D_{\mathrm{keyfmt}}}
\newcommand{\Thref}{\mathcal T_{\mathrm{href}}}
\newcommand{\Thself}{\mathcal T_{\mathrm{hself}}}
\newcommand{\Tct}{\mathcal T_{\mathrm{ct}}}
\newcommand{\Tdiff}{\mathcal T_{\mathrm{diff}}}
\newcommand{\Tkey}{\mathcal T_{\mathrm{key}}}
\newcommand{\mlkem}{\textsf{ML-KEM}}
\newcommand{\hqc}{\textsf{HQC}}
\newcommand{\cdone}{\textsf{CD1}}
\newcommand{\cdstar}{\textsf{CD*}}

\newcommand{\mlkemnative}{\textsf{mlkem-native}}
\newcommand{\algobox}[2]{%
	\fbox{%
		\begin{minipage}[t]{#1}
			\vspace{0.35em}
			#2
			\vspace{0.35em}
\end{minipage}}}
\newcommand{\algoline}[2]{\texttt{#1} & #2\\}

\usepackage[backend=biber,style=numeric]{biblatex}
\makeatletter

\newcommand{\@toptitlebar}{%
  \hrule height 2pt
  \vskip 0.25in
  \vskip -\parskip
}

\newcommand{\@bottomtitlebar}{%
  \vskip 0.29in
  \vskip -\parskip
  \hrule height 2pt
  \vskip 0.09in
}

\renewcommand{\maketitle}{%
  \par
  \begingroup
    \thispagestyle{plain}
    \begin{center}
      \vspace*{0.1in}
      \@toptitlebar
      {\LARGE\scshape \@title\par}
      \@bottomtitlebar
      \vspace{0.15in}
      {\normalsize\bfseries \@author\par}
      \vspace{0.2in}
    \end{center}
  \endgroup
}

\renewcommand{\section}{%
  \@startsection{section}{1}{\z@}%
                {-2.0ex \@plus -0.5ex \@minus -0.2ex}%
                {1.2ex \@plus 0.2ex}%
                {\large\bfseries\raggedright}}

\renewcommand{\subsection}{%
  \@startsection{subsection}{2}{\z@}%
                {-1.8ex \@plus -0.5ex \@minus -0.2ex}%
                {0.8ex \@plus 0.2ex}%
                {\normalsize\bfseries\raggedright}}

\renewcommand{\subsubsection}{%
  \@startsection{subsubsection}{3}{\z@}%
                {-1.5ex \@plus -0.5ex \@minus -0.2ex}%
                {0.5ex \@plus 0.2ex}%
                {\normalsize\bfseries\raggedright}}

\makeatother

\title{What Can Verifiable Decapsulation Tests Certify? \\
	Pass Bounds and Fault-Recognition Limits \\ for FO-Based KEMs}

\author{%
José Luis Delgado \\
\small Universitat Oberta de Catalunya \\
\texttt{jdelgado13@uoc.edu}
}

\begin{document}

\maketitle

\begin{abstract}
	Black-box tests for Fujisaki--Okamoto decapsulation observe the sampled
	execution seen by the harness, whereas the reencryption computation itself is
	visible only through the values that reach final key derivation.  We study
	confirmation-code-augmented KEM variants under an honest-reference harness in
	which the reference encapsulation fixes a hidden final-key point
	\(\langle\mathsf{good},B,W\rangle\), with \(W\) the confirmation witness.  For a
	\(q\)-localized system under test, acceptance is bounded by honest correctness
	error, adversarial aliasing, final-key freshness defects, a hit on the
	localized suffix list \(Q_G(B)\), and \(2^{-\kappa}\).  A one-query construction
	from any predictor of \(W\) matches this bound up to the fresh-key coincidence
	term, so the list-hit event is the black-box obstruction measured by the
	harness.
	
	The list-hit term is bounded either by a cUP-faithful harness certificate,
	which transfers source confirmation-code unpredictability with a \(q\)-loss, or
	by an average conditional min-entropy bound, with separate RawEnt and TailEnt
	hypotheses for short diagnostic and truncation-tail codes.  The same model
	proves a dependency-cone lower bound for non-certification claims.  When the black-box
	observation of an honest-support harness factors through the
	confirmation-observable final-key target, every operation outside the
	support-active cone has a coupled erasure implementation with the same
	transcript distribution; over any implementation class containing that erasure,
	soundness and completeness errors of an execution certifier satisfy
	\(\alpha+\beta\ge 1\).  The \mlkem{} and \hqc{} case studies distinguish
	theorem-covered positive rows, finite-catalog artifact rows, and
	non-certification rows that carry a cone-inactivity certificate.  The security
	of the standard KEM lines is the construction-level security supplied by the
	cited source analyses.
\end{abstract}

\section{Introduction}

Post-quantum KEMs built from the Fujisaki--Okamoto transform rely on a
decapsulation check that is essential for security.  A decapsulation algorithm
decrypts the received ciphertext, rederives the encryption randomness,
recomputes the ciphertext, compares the result to the received ciphertext, and
then derives the final shared secret.  The check is security-critical, but an
ordinary input-output test observes only the sampled path.  On honestly
generated ciphertexts, an implementation that omits part of the check may still
return the expected shared secret.

Glabush, G{\"u}nther, H{\"o}velmanns, and Stebila introduced verifiable
decapsulation for FO-based KEMs~\cite{Glabush2025VerifiableDecapsulation}.
Their construction adapts the confirmation-code methodology of Fischlin and
G{\"u}nther~\cite{FischlinGuenther2023VV}: encryption produces an internal
confirmation value, decapsulation recomputes it, and the value is tied to the
subsequent key material.  In the KEM setting, this makes the relevant
reencryption computation part of the ordinary shared-key outcome: a fault that
misses the computation produces the wrong final input except on the usual
guessing events.

After the final key derivation includes a confirmation value, the pass
probability measured by a test depends on the augmented KEM variant, the fault
class, the peer model, and the sampled input distribution.  Removing the
confirmation value from the final key input, changing branch selection on
malformed ciphertexts, and applying the same endpoint mutation on both sides
give different harness-indexed statements.  For a fixed KEM variant, fault
class, and harness, our theorem bounds the corresponding pass probability
through a localized final-key query list.  The source notions are the cUP and
fCOR games of verifiable decapsulation; the harness layer reformulates them as a
list-cUP testing statement with an explicit SUT view and support.

\paragraph{Main results.}
The central result is a formal theory of confirmation-code testing as
list-cUP evidence under harness-indexed views.  A canonical honest-reference
instance fixes a hidden final-key point
\[
x^\star=\langle \mathsf{good},B,W\rangle
\]
and a reference key \(K=G_K(x^\star)\), where \(B\) is the localized
transcript prefix and \(W\) is the hidden confirmation witness.  A tested
decapsulator may be randomized, adaptive, stateful, and may know the public
data and the decapsulation input prescribed by the harness.  Its localized
final-key queries define a list
\[
Q_G(B)=\{w:\mathcal S \text{ queried }
G_K(\langle\mathsf{good},B,w\rangle)\}.
\]
If adversarial aliases and final-key freshness defects are accounted for, then
\[
\mathsf{Pass}^{\mathsf{href}}_{\Pi^C,\mathcal H}(\mathcal S)
\le
\delta_{\Pi^C}
+\epsilon_{\mathsf{alias}}
+\epsilon_{\mathsf{fresh}}
+\Pr[W\in Q_G(B)]
+2^{-\kappa}.
\]
Thus every meaningful positive testing claim reduces to bounding the list-hit
term.  The reduction achieves black-box tightness in the sense that, for every single-candidate
predictor of the hidden witness from the SUT view, there exists a (1)-localized SUT whose honest-reference pass probability is exactly the predictor’s success probability 
plus the fresh-key coincidence term. The rest of the
paper supplies the ways to bound this unavoidable term: source cUP under a
cUP-faithful harness certificate, conditional entropy, and diagnostic
RawEnt/TailEnt assumptions for short or tail-derived codes.

To bound the list-hit term, the source-security route defines a harness
list-cUP game and proves that a \(q\)-candidate list adversary loses a factor
\(q\) relative to a single-candidate cUP bound.  Under a source bound
\(\epsilon_{\mathsf{cUP}}\), this yields
\[
\mathsf{Pass}^{\mathsf{href}}
\le
\delta_{\Pi^C}
+\epsilon_{\mathsf{alias}}
+\epsilon_{\mathsf{fresh}}
+q\epsilon_{\mathsf{cUP}}
+2^{-\kappa}.
\]
The entropic route proves that if \(W\) has average conditional min-entropy
\(\lambda\) given a view that determines the candidate list, then
\[
\mathsf{Pass}^{\mathsf{href}}
\le
\delta_{\Pi^C}
+\epsilon_{\mathsf{alias}}
+\epsilon_{\mathsf{fresh}}
+q2^{-\lambda}
+2^{-\kappa},
\]
with an explicit smoothing loss when only smooth min-entropy holds.

Because the sampled peer and support are part of the experiment, the harness
fixes the interpretation of these bounds.  Honest-self tests can accept a
self-consistent endpoint that erases the code from both encapsulation and
decapsulation, even though the same endpoint fails against a sound reference
peer except with correctness error and final-key guessing.
Support-equivalent endpoints have the same pass probability only on the tested
support.  A retained range of \(L\) bits is only a capacity bound and implies no
lower bound on conditional min-entropy without a separate unpredictability
claim.  The testing theorems in this paper are classical ROM statements; QROM
security claims are inherited only from the cited construction-level
verifiable-decapsulation and FO analyses.

\paragraph{Contributions.}
\begin{enumerate}[leftmargin=*]
	\item We define honest-reference and honest-self testing experiments for
	confirmation-code-augmented FO decapsulation, including canonical final-key
	inputs, hidden reference evaluations, localized final-key query lists,
	adversarial aliasing, final-key freshness, and correctness bad events.
	\item We prove a tight list-hit detectability theorem.  Every \(q\)-localized
	SUT satisfies
	\[
	\Pass^{\mathsf{href}}
	\le
	\delta+\epsilon_{\mathsf{alias}}+\epsilon_{\mathsf{fresh}}
	+\Pr[W\in Q_G(B)]+2^{-\kappa}.
	\]
	Moreover, the list-hit term is unavoidable: any single-candidate predictor of
	\(W\) from the SUT view induces a one-query SUT whose pass probability is the
	predictor success probability plus fresh-key coincidence.
	\item We prove the dependency-cone lower bound for black-box honest-reference
	certification.  If the observation of an honest-support harness factors
	through the confirmation-observable final-key point, then no operation outside
	the support-active dependency cone of that point can be certified by a
	black-box test over an implementation class containing the coupled erasure.
	Equivalently, over erasure-closed classes, every such operation admits an
	erasure fault with identical transcript distribution, so nontrivial execution
	certification is information-theoretically impossible.
	\item We prove the two positive routes that meet this lower bound.  The
	source-security route defines a harness list-cUP game with a \(q\)-loss
	from ordinary cUP under a cUP-faithful harness certificate.  The entropic route
	derives a \(q2^{-\lambda}\) list-hit bound from average conditional min-entropy
	of the witness, including smoothing.
	\item We prove a diagnostic-code theorem for truncated hashes of hidden raw
	witnesses.  This theorem covers short artifact codes such as
	\mlkem{}-\cdone{}; those variants use the source cUP bound only when a
	source-equivalence certificate is supplied.
	\item We map the dependency cone to \mlkem{} and the August 2025 salted-SFO
	\hqc{} line.
	For the source \mlkem{} code with \(|S|=2\), only \(2(k+1)\) selected
	uncompressed coefficients enter the source-code cone; the remaining
	\(254(k+1)\) coefficients are outside that code cone unless exposed by a
	separate harness target.  For \hqc{}-\cdone{}, the retained tail cone contains
	\(5,8,8\) bits for HQC-1/3/5, leaving \(0,3,29\) tail bits outside the
	\cdone{} cone; \cdstar{} retains the full \(5,11,37\)-bit tail range.
	\item We provide a reproducible artifact with pinned baselines, derived
	variants, harnesses, mutant catalogs, frozen JSONL records, three compact
	summaries, and two plots.  The experiments are
	interpreted as witnesses for theorem-relevant mutant behavior, with
	construction-level security for the standard KEM lines supplied by the cited
	source analyses.  The dependency-cone theorem upgrades a negative row only
	when it is paired with a cone-inactivity certificate for the full
	confirmation-observable target.  Decision-only, component-selective,
	non-retained-tail, and unselected-coefficient rows are reported either as
	finite-catalog observations, as source-code-cone boundaries, or as
	theorem-covered non-certification results according to the certificate
	available for the stated honest-support harness.
\end{enumerate}

\paragraph{Instantiations and artifact.}
The model is instantiated on two augmented families: for
\mlkem{}, we use the confirmation-code line of the 2025
verifiable-decapsulation paper and evaluate an artifact diagnostic variant over
the FIPS 203 algorithms; for the August 2025 public \hqc{}
specification~\cite{HQC2025Spec}, the salted flow yields a direct transcript
\[
H(ek_{\mathsf{KEM}})\|m\|\mathsf{salt}
\]
and our augmented variants take their confirmation code from the truncation
tail.  We study
\cdone{}, which retains one byte when available, and \cdstar{}, which retains
the full truncation-tail range budget.

The accompanying artifact, FO-FaultBench, available at \url{https://github.com/hypergalois/FO-FaultBench}, 
pins the \mlkemnative{} and August
2025 \hqc{} baselines, materializes the transformed variants, and runs the
harnesses used by the theorem statements.  The experiments check the stated
harness behavior on concrete code: binding faults are detected under
\(\Dhonref\), decision-only faults are support-dependent, symmetric endpoint
faults separate \(\Dhonself\) from \(\Dhonref\), and probabilistic code
guessing follows the retained-code reference scale for the tested mutants.

\section{Related Work}
\label{sec:related}

Glabush, G{\"u}nther, H{\"o}velmanns, and Stebila introduced verifiable decapsulation
for FO-based KEMs~\cite{Glabush2025VerifiableDecapsulation}.  Their paper
adapts confirmation codes to FO-based KEMs and proves that, for the targeted
faulty-but-benign implementations, missing reencryption can appear as a
correctness failure.  The construction builds on the verifiable-verification
methodology of Fischlin and G{\"u}nther~\cite{FischlinGuenther2023VV}, where a
verification computation produces disposable confirmation information and ties
that information to later protocol functionality.  We keep the same
construction and analyze the pass probability measured by a specified harness.

A concrete test harness may allow several final-key queries, run a
self-consistent mutant as the peer, sample a restricted support, or use
diagnostic codes shorter than the source confirmation code.  The harness
list-cUP, conditional min-entropy, support separation, and diagnostic-code
theorems below formalize these differences as certification signals measured
by that harness.

The FO-transform background used here is the modular analysis of Hofheinz,
H{\"o}velmanns, and Kiltz~\cite{HHK2017}, together with later work on prefix
hashing and public-key binding, which explains why FO inputs often include
public-key-derived material in multi-user
settings~\cite{DumanEtAl2021PrefixHashing}.  Work on decryption failures,
explicit rejection, and failing plaintexts clarifies the distinction between
the reencryption computation and the accept-or-reject behavior of the
decapsulation oracle~\cite{HovelmannsHulsingMajenz2022FailingGracefully,
	HovelmannsMajenz2023FailingGracefullyNote,
	HovelmannsKudinov2025ExplicitImplicit}.  Recent work on alternatives to
reencryption through rigidity and range-check style structure also separates
the reencryption computation from the rejection
interface~\cite{HovelmannsHulsingMajenzSisinni2025ReEncryptionAlternatives}.
The same separation appears in the harness statements below, where
recomputation faults and branch-selection faults have different supports.

The \hqc{} instantiation uses the salted FO analysis of Glabush,
H{\"o}velmanns, and Stebila~\cite{GlabushHoevelmannsStebila2025MultiTarget},
which formalizes the public salt used to mitigate multi-target collision
attacks and ties the randomness and key derivations to public-key material, the
message, and the salt.  The public \hqc{} specification dated 2025-08-22 adopts
this salted implicit-rejection line and explicitly uses
\(H(ek_{\mathsf{KEM}})\|m\|\mathsf{salt}\) in the derivation of the key and
reencryption randomness~\cite{HQC2025Spec}.  The experiments use the
August 2025 specification as the main \hqc{} object, with the earlier February
line retained for drift checks.

For confirmation-code harnesses, the dependency-cone lower bound specializes
observational equivalence to the declared target vector.  Once the black-box
observation factors through that vector, operations that are inactive for the
target can be erased without changing the transcript distribution.  The
contribution is the target/cone grammar for confirmation-code testing, not the
abstract indistinguishability fact by itself; the certificates below record the
support, target, and implementation-class conditions under which a negative
row becomes a non-certification theorem.

\mlkem{} is standardized in FIPS~203~\cite{NISTFIPS203}, and NIST SP 800-227 contains
recommendations for implementing and using KEMs~\cite{NISTSP800227}.  The
claims below study the implementation signal created by confirmation-code
binding in augmented variants derived from those baselines.

For the augmented line, construction-level provenance comes from the QROM
security literature, including double-sided and measure-rewind styles of
one-way-to-hiding, online and oracle-masked extraction, direct Kyber analyses,
and FO lifting
results~\cite{BindelEtAl2019TighterQROM,JiangZhangMa2019ExplicitRejection,
	ShanGeXue2022QCCA,DonEtAl2021OnlineExtractability,
	DingEtAl2022KyberInjectivity,BarbosaHulsing2023KyberFO,
	GeShanXue2023ExplicitRejection,GeShanXue2023FOQCCALift,
	GeLiaoXue2024MRE,MajenzSisinni2024FFPNG}.  These results provide provenance for
the standard FO security games cited in \Cref{sec:source-security}.  The
testing statements in this paper are classical ROM statements about pass
probabilities under named testing harnesses.

\section{Model}
\label{sec:model}

\subsection{Canonical final-key point}

We model the testing target as a confirmation-code-augmented FO KEM
\(\Pi^C\) with algorithms \(\KeyGen,\Encaps^C,\Decaps^C\).  All encodings in
the model are length-delimited, domain-separated, injective, and prefix-free
unless an explicit aliasing error is stated.  The final-key oracle \(G_K\) is a
classical lazy-sampled random oracle returning \(\kappa\)-bit strings.

\begin{definition}[Canonical confirmation-code FO test line]
	\label{def:canonical-instance}
	A canonical honest-reference instance is a tuple
	\[
	\Inst=(ek,dk,c,K,B,W,x^\star)
	\]
	sampled as follows:
	\[
	(ek,dk)\leftarrow \KeyGen(1^\lambda),\qquad
	(K,c,\tau)\leftarrow \Encaps^C(ek).
	\]
	The encapsulation transcript \(\tau\) determines a localized transcript prefix
	\[
	B=B(\tau)\in\bits^\ast,
	\qquad
	W=W(\tau)\in\mathcal W,
	\qquad
	x^\star=\langle \mathsf{good},B,W\rangle .
	\]
	The honest accepted key is
	\[
	K=G_K(x^\star)\in\bits^\kappa .
	\]
	If the final-key encoding is not syntactically canonical, any adversarial
	collision with the canonical point is charged to the aliasing defect in
	\Cref{def:alias-freshness}.
	The honest correctness error is
	\[
	\delta_{\Pi^C}
	=
	\Pr[\Decaps^C(dk,c)\neq K].
	\]
	The harness withholds \(K\), \(W\), and \(x^\star\) from the tested system
	except through declared public data and oracle access.
\end{definition}

\begin{definition}[Adversarial aliasing and final-key freshness]
	\label{def:alias-freshness}
	Let \(\Can\) be the canonical parser used for final-key oracle inputs, with
	undefined inputs mapped to \(\bot\).  For a run of a SUT \(\mathcal S\), let
	\(\mathcal Q_K^{\mathsf{all}}\) be the set of all inputs queried by
	\(\mathcal S\) to \(G_K\).  The adversarial aliasing event is
	\[
	\Alias_{\mathcal S}
	=
	\left[
	\exists x\in\mathcal Q_K^{\mathsf{all}}\setminus\{x^\star\}:
	\Can(x)=\Can(x^\star)
	\right].
	\]
	Set
	\[
	\epsilon_{\mathsf{alias}}(\mathcal S,\mathcal H)
	=
	\Pr[\Alias_{\mathcal S}].
	\]
	For a SUT class \(\mathcal C\), write
	\[
	\epsilon_{\mathsf{alias}}(\mathcal C,\mathcal H)
	=
	\sup_{\mathcal S\in\mathcal C}
	\epsilon_{\mathsf{alias}}(\mathcal S,\mathcal H).
	\]
	
	Let \(\mathcal E\) be the event that honest correctness holds, no adversarial
	alias occurs, and the hidden point \(x^\star\) is not queried.  Let \(V\) be
	the complete SUT view at output time, including all declared oracle replies
	except \(G_K(x^\star)\).  If \(\Pr[\mathcal E]=0\), set
	\(\epsilon_{\mathsf{fresh}}(\mathcal S,\mathcal H)=0\).  Otherwise, the
	final-key freshness defect is
	\[
	\epsilon_{\mathsf{fresh}}(\mathcal S,\mathcal H)
	=
	\max\left\{0,\,
	\mathbb E_{v\leftarrow V\mid\mathcal E}
	\left[
	\max_{k\in\bits^\kappa}
	\Pr[K=k\mid V=v,\mathcal E]
	\right]
	-
	2^{-\kappa}
	\right\}.
	\]
	For a class \(\mathcal C\), define
	\[
	\epsilon_{\mathsf{fresh}}(\mathcal C,\mathcal H)
	=
	\sup_{\mathcal S\in\mathcal C}
	\epsilon_{\mathsf{fresh}}(\mathcal S,\mathcal H).
	\]
	The harness is secret-separating for \(\mathcal S\) if
	\(\epsilon_{\mathsf{fresh}}(\mathcal S,\mathcal H)=0\).  It is
	\((\epsilon_{\mathsf{alias}},\epsilon_{\mathsf{fresh}})\)-well formed for a
	class \(\mathcal C\) if the corresponding suprema over \(\mathcal C\) are
	bounded by those values.
\end{definition}

In the salted \hqc{} presentation of \Cref{sec:transform}, the scheme-specific
prefix is \(B=H(ek_{\mathsf{KEM}})\|m\|\mathsf{salt}\), whereas nonsalted
lines use an empty or absent salt field.  The witness \(W\) may be a source
confirmation code, a tail-derived code, or a diagnostic code; its cryptographic
interpretation depends on the theorems below.

\subsection{Honest-reference testing}

\begin{definition}[Honest-reference harness]
	\label{def:sut-interface}
	\label{def:href}
	For a possibly randomized, stateful system under test \(\mathcal S\), the
	experiment \(\HRef_{\Pi^C,\mathcal H}(\mathcal S)\) runs:
	\begin{enumerate}[leftmargin=2em]
		\item sample \(\Inst\) as in \Cref{def:canonical-instance};
		\item give \(\mathcal S\) the decapsulation input prescribed by the harness,
		typically \((dk,c,ek,\mathsf{aux})\), but not \(K\), \(W\), or \(x^\star\);
		\item answer oracle queries using the same domain-separated random oracles as
		the reference line, while the reference evaluation \(G_K(x^\star)\) is
		hidden unless \(\mathcal S\) itself queries \(x^\star\);
		\item receive an output \(\widehat K\) from \(\mathcal S\);
		\item output \(1\) iff \(\widehat K=K\).
	\end{enumerate}
	Define
	\[
	\mathsf{Pass}^{\mathsf{href}}_{\Pi^C,\mathcal H}(\mathcal S)
	=
	\Pr[\HRef_{\Pi^C,\mathcal H}(\mathcal S)=1].
	\]
	For a realized run, the localized final-key candidate set is
	\[
	Q_G(B)
	=
	\{\,w\in\mathcal W:
	\mathcal S \text{ queried }
	G_K(\langle\mathsf{good},B,w\rangle)\,\}.
	\]
	The tested system is \(q_G\)-localized if \(|Q_G(B)|\le q_G\) in every run.
	The central hit event is
	\[
	\mathsf{Hit}=[W\in Q_G(B)].
	\]
	Leakage is accounted for by the final-key freshness defect of
	\Cref{def:alias-freshness}.  This captures accidental exposure of the hidden
	key point through instrumentation, debug logs, shared oracle tables,
	deterministic seeds, transcript caches, or any other channel in the SUT view.
	A certified harness must give explicit reasons why
	\(\epsilon_{\mathsf{fresh}}=0\), or must carry the stated freshness defect in
	the bound.
\end{definition}

The harness model permits a broad tested system, the SUT may know \(dk\),
\(ek\), \(c\), the salt, and any auxiliary data explicitly declared by the
harness, and it may be adaptive and stateful.  The hidden confirmation witness
\(W\) and the hidden reference key evaluation are excluded from the SUT input,
except through the final-key query whose absence is being tested.

\begin{definition}[Pass and detection probability]
	For any harness \(\mathcal T=(D,V)\), define
	\[
	\Pass_{\Pi,\mathcal S,\mathcal T}
	=
	\Pr_{x\leftarrow D}[V(x,\mathcal S)=1],
	\qquad
	\Det_{\Pi,\mathcal S,\mathcal T}=1-\Pass_{\Pi,\mathcal S,\mathcal T}.
	\]
	For a class \(\mathcal C\), write
	\[
	\Pass_{\Pi,\mathcal C,\mathcal T}
	=
	\sup_{\mathcal S\in\mathcal C}\Pass_{\Pi,\mathcal S,\mathcal T}.
	\]
\end{definition}

\subsection{Harness families}

Across the theorem statements and artifact, five harness families distinguish
the honest-reference signal from self-consistency, malformed-input,
differential, and key-format behavior.  The positive theorem is about
\(\Thref\); the other harnesses mark boundaries of what that test can certify.

\begin{table}[H]
	\centering
	\footnotesize
	\begin{tabularx}{\textwidth}{@{}>{\raggedright\arraybackslash}p{2.0cm}
			>{\raggedright\arraybackslash}p{3.7cm}
			>{\raggedright\arraybackslash}p{3.2cm}
			>{\raggedright\arraybackslash}X@{}}
		\toprule
		Harness & Sampled object & Peer model & Main use in the paper \\
		\midrule
		\(\Thref\) & reference-generated \((ek,dk,c,K,B,W)\) & tested decapsulation against a sound peer & list-cUP detectability and binding mutants \\
		\(\Thself\) & tested endpoint's own \((ek,dk,c,K)\) & faulty endpoint talks to itself & self-vs-reference separation \\
		\(\Tct^\mu\) & malformed ciphertext under mode \(\mu\) & reference-vs-tested behavior on indexed support & branch-only and component-selective faults \\
		\(\Tdiff^\mu\) & same serialized input for both implementations & direct differential comparison & mode-selective signatures and drift checks \\
		\(\Tkey^\nu\) & key or layout perturbation under mode \(\nu\) & parser and layout comparison & supplementary artifact coverage \\
		\bottomrule
	\end{tabularx}
	\caption{Harnesses used in the theorem statements.}
	\label{tab:harnesses}
\end{table}

\begin{definition}[Honest-self harness]
	\label{def:hself}
	The experiment \(\HSelf_{\Pi',\mathcal H}(\mathcal S)\) samples key pairs,
	ciphertexts, and reference keys from the same endpoint \(\Pi'\) implemented by
	the SUT family.  It accepts iff the decapsulated key returned by the tested
	endpoint agrees with that endpoint's own encapsulated key.  This differs from
	\(\HRef_{\Pi^C,\mathcal H}\), which samples from the canonical reference
	\(\Pi^C\) and compares against \(G_K(\langle\mathsf{good},B,W\rangle)\).
\end{definition}

\begin{definition}[Support-equivalent SUTs]
	\label{def:support-equivalent}
	Let \(\mathcal D\) be a harness distribution over inputs
	\(z=(ek,dk,c,\mathsf{aux})\).  Two randomized SUTs
	\(\mathcal S_0,\mathcal S_1\) are \(\Delta\)-support-equivalent on
	\(\mathcal D\) if
	\[
	\Delta_{\mathsf{TV}}
	\bigl((z,\mathcal S_0(z)),(z,\mathcal S_1(z))\bigr)
	\le \Delta
	\]
	for \(z\leftarrow\mathcal D\).
\end{definition}

\begin{definition}[Reset-product testing]
	\label{def:reset-product}
	A \(t\)-fold reset-product harness samples independent instances
	\(\Inst_1,\ldots,\Inst_t\), uses independent domain-separated oracle
	namespaces \(G_K^{(1)},\ldots,G_K^{(t)}\), resets the SUT state between runs,
	and accepts iff all \(t\) runs accept.
\end{definition}

\subsection{Fault vocabulary}

Faults are named by the operation they omit or alter:
\[
\mathcal O =
\{O_\theta,O_{\mathsf{reencrypt}},O_{\mathsf{code}},
O_{\mathsf{compare}},O_{\mathsf{reject}},O_{\mathsf{bind}},
O_{\mathsf{select}}\}.
\]
Here \(O_\theta\) and \(O_{\mathsf{reencrypt}}\) refer to deriving
reencryption randomness and recomputing the ciphertext, \(O_{\mathsf{code}}\)
to producing the confirmation code, \(O_{\mathsf{bind}}\) to including the
mandated transcript inputs in the final key derivation, and the other entries
to comparison, rejection, and key selection.

\begin{table}[t]
	\centering
	\footnotesize
	\begin{tabularx}{\textwidth}{@{}>{\raggedright\arraybackslash}p{2.6cm}
			>{\raggedright\arraybackslash}p{4.0cm}
			>{\raggedright\arraybackslash}X@{}}
		\toprule
		Class & Typical examples & Main harness behavior \\
		\midrule
		Recomputation faults & skip reencryption, recompute only part of the ciphertext, overwrite the confirmation code & visible under \(\Thref\) only when the omitted work is covered by the code or the ordinary key path \\
		Binding faults & omit \(W\), omit \(H(ek_{\mathsf{KEM}})\), omit \(\mathsf{salt}\), confuse transcript inputs & visible under \(\Thref\) when they miss the hidden final-key point \\
		Decision faults & always select \(K_{+}\), ignore rejection, compare only one component & often invisible on honest support and visible only under malformed or differential modes \\
		Parsing and format faults & wrong key offset, wrong serialized layout, missing salt field & studied through \(\Tkey\) and \(\Tdiff\) \\
		Symmetric endpoint faults & apply the same binding mutation to both encapsulation and decapsulation & can pass \(\Thself\) and fail \(\Thref\) \\
		\bottomrule
	\end{tabularx}
	\caption{Fault classes used in the paper. The positive theorem concerns
		localized final-key list hits; operation coverage requires a separate
		code-source argument.}
	\label{tab:fault-classes}
\end{table}

\section{Split Presentation of the Augmented Line}
\label{sec:transform}

For the artifact and proofs, we use a split presentation of the
confirmation-code-augmented FO construction.  The presentation separates the
randomness derivation used for reencryption from the final key derivation,
matches the salted \hqc{} syntax, and places the confirmation code at one final
random-oracle point.  The reject-key expression uses the standard rejection
secret \(\sigma\) stored in the decapsulation key.

\begin{figure}[H]
	\centering
	\begin{minipage}[t]{0.47\linewidth}
		\centering
		\algobox{0.95\linewidth}{%
			\textbf{Encapsulation}\par\medskip
			\footnotesize
			\begin{tabular}{@{}r >{\raggedright\arraybackslash}p{0.75\linewidth}@{}}
				\algoline{01}{sample \(m\) and \(\mathsf{salt}\)}
				\algoline{02}{set \(B\gets H(ek_{\mathsf{KEM}})\|m\|\mathsf{salt}\)}
				\algoline{03}{set \(\theta\gets G_\theta(B)\)}
				\algoline{04}{compute \((c_{\mathsf{PKE}},W)\gets \EncC(ek_{\mathsf{KEM}},m;\theta)\)}
				\algoline{05}{set \(c_{\mathsf{KEM}}\gets(c_{\mathsf{PKE}},\mathsf{salt})\)}
				\algoline{06}{set \(K\gets G_K(\langle\mathsf{good},B,W\rangle)\)}
				\algoline{07}{return \((K,c_{\mathsf{KEM}})\)}
		\end{tabular}}
	\end{minipage}\hfill
	\begin{minipage}[t]{0.47\linewidth}
		\centering
		\algobox{0.95\linewidth}{%
			\textbf{Decapsulation}\par\medskip
			\footnotesize
			\begin{tabular}{@{}r >{\raggedright\arraybackslash}p{0.75\linewidth}@{}}
				\algoline{01}{parse \(c_{\mathsf{KEM}}=(c_{\mathsf{PKE}},\mathsf{salt})\)}
				\algoline{02}{compute \(m'\gets\Dec(dk_{\mathsf{KEM}},c_{\mathsf{PKE}})\)}
				\algoline{03}{set \(B'\gets H(ek_{\mathsf{KEM}})\|m'\|\mathsf{salt}\)}
				\algoline{04}{set \(\theta'\gets G_\theta(B')\)}
				\algoline{05}{compute \((c'_{\mathsf{PKE}},W')\gets\EncC(ek_{\mathsf{KEM}},m';\theta')\)}
				\algoline{06}{set \(K_{+}'\gets G_K(\langle\mathsf{good},B',W'\rangle)\)}
				\algoline{07}{set \(K_{-}\gets J(H(ek_{\mathsf{KEM}})\|\sigma\|c_{\mathsf{KEM}})\)}
				\algoline{08}{return \(\ctsel(c'_{\mathsf{PKE}}=c_{\mathsf{PKE}},K_{+}',K_{-})\)}
		\end{tabular}}
	\end{minipage}
	\caption{Simplified split presentation used in the paper.  The final good key
		is evaluated at a canonical point \(\langle\mathsf{good},B,W\rangle\), where
		\(B\) is the transcript prefix and \(W\) is the recomputed confirmation
		witness.}
	\label{fig:split-alg}
\end{figure}

\begin{definition}[Split presentation of verifiable salted FO]
	Let \(\mathsf{PKE}^C=(\KeyGen,\EncC,\Dec)\) be a PKE scheme whose encryption
	algorithm outputs a ciphertext and a confirmation code.  The artifact line,
	denoted \(\VSFO[\mathsf{PKE}^C,G_\theta,G_K,H,J]\), is the KEM in
	\Cref{fig:split-alg}.  The ciphertext is
	\[
	c_{\mathsf{KEM}}=(c_{\mathsf{PKE}},\mathsf{salt}),
	\]
	the reencryption randomness is derived as
	\[
	\theta=G_\theta(H(ek_{\mathsf{KEM}})\|m\|\mathsf{salt}),
	\]
	and the good key is derived as
	\[
	K_{+}=G_K(\langle\mathsf{good},H(ek_{\mathsf{KEM}})\|m\|\mathsf{salt},cd\rangle).
	\]
\end{definition}

In the split presentation, \(G_\theta\) controls the deterministic
reencryption path and \(G_K\) receives the recomputed code after reencryption,
which are the two oracle points used by the testing statement.  A fault that
misses \(cd\) lacks the honest final good-key query; the FO comparison and
rejection logic is unchanged.

\begin{remark}[Relation to the 2025 construction]
	The construction follows the same confirmation-code principle as
	\cite{Glabush2025VerifiableDecapsulation}.  Source security terms from that
	work are transported to this split presentation only when the
	source-equivalence certificate of \Cref{def:source-equivalent-split} is
	satisfied.  Diagnostic variants that change or truncate the witness are
	handled by the diagnostic entropy route instead.  The explicit split matches
	the August 2025 \hqc{} salted transcript
	\(H(ek_{\mathsf{KEM}})\|m\|\mathsf{salt}\) and the canonical final point
	\(G_K(\langle\mathsf{good},B,W\rangle)\) used in the testing theorem.
\end{remark}

\begin{definition}[Source-equivalent split presentation]
	\label{def:source-equivalent-split}
	Let \(\Pi^{\mathsf{src}}\) be the source confirmation-code transform and let
	\(\Pi^{\mathsf{split}}\) be a split presentation with separate \(G_\theta\)
	and \(G_K\) domains.  The split presentation is source-equivalent for a SUT
	class \(\mathcal C\) and harness \(\mathcal H\) if there is an efficiently
	computable bijection \(\Phi\) on transcripts such that:
	\begin{enumerate}[label=(S\arabic*),leftmargin=2em]
		\item \((ek,dk,c,B,W)\) has the same distribution in
		\(\Pi^{\mathsf{src}}\) and \(\Pi^{\mathsf{split}}\) after applying \(\Phi\).
		\item The source final-key oracle input and the split final-key input are
		related by a domain-separated injective relabeling, so random-oracle tables
		can be coupled by renaming inputs.
		\item The \(G_\theta\) queries of the split presentation are simulatable from
		the source cUP interface and do not reveal any additional function of \(W\)
		beyond the source transcript.
		\item The accept predicate and honest-reference key equality event coincide
		after the same transcript relabeling.
	\end{enumerate}
\end{definition}

\begin{proposition}[Transfer across a source-equivalent split]
	\label{prop:split-transfer}
	If \(\Pi^{\mathsf{split}}\) is source-equivalent to
	\(\Pi^{\mathsf{src}}\) in the sense of \Cref{def:source-equivalent-split},
	then cUP, list-cUP, and fCOR advantages for the two presentations are equal up
	to the explicit relabeling of oracle domains.  In particular, a source cUP
	bound may be used in \Cref{thm:certified-lcup} for the split presentation only
	under this source-equivalence certificate.
\end{proposition}

\begin{proof}
	Run the source experiment and the split experiment with the same sampled
	transcript under the bijection \(\Phi\).  Couple their random oracles by the
	injective domain relabeling from condition (S2), and answer all \(G_\theta\)
	queries using the simulator from (S3).  Conditions (S1) and (S4) give
	identical distributions of \((B,W)\) and identical acceptance events.
	Therefore any adversary in one presentation is transformed into an adversary
	in the other with the same success probability, and conversely.
\end{proof}

\begin{remark}[Diagnostic splits are not source transfers]
	\label{rem:diagnostic-not-source-transfer}
	The \mlkem{}-\cdone{} and \hqc{}-\cdone{}/\cdstar{} artifact lines are not
	asserted to be source-equivalent to the source \mlkem{} confirmation-code
	construction.  They change the witness source or shorten it diagnostically.
	Their interpretation therefore goes through the hashed-witness theorem,
	\mlkem{}-\cdone{} RawEnt, or \hqc{} TailEnt, not through direct transfer of
	the source cUP theorem.
\end{remark}

\paragraph{Testing split.}
In an FO implementation, randomness derivation, good-key derivation, and
transcript binding can be close together in the code and in the proof syntax.
For the present testing question, the reencryption randomness and the final key
derivation are separate oracle points.  The confirmation code is used after
reencryption has produced the value bound into the final key.
A fault that skips a code-covered computation and therefore fails to determine
\(W\) loses the exact value needed at the final key point; operations outside
the witness dependency cone require separate harness or code-source evidence.

\paragraph{Rejection path.}
The ordinary rejection structure is unchanged: decapsulation decrypts,
recomputes, compares, and selects between a good key and a reject key.  The
support-coincidence theorem therefore applies to the transformed line as stated.
A branch-only fault that agrees with correct decapsulation on all honestly
generated ciphertexts is a malformed-support question.

\section{Detectability Results}
\label{sec:theory}

\subsection{Honest-reference list-hit bound}

\begin{theorem}[Honest-reference list-hit bound with adversarial bad events]
	\label{thm:code-blind}
	\label{thm:list-hit-bad-events}
	Let \(\Pi^C\) be a canonical confirmation-code FO line with final key length
	\(\kappa\) and correctness error \(\delta_{\Pi^C}\).  For any possibly
	randomized, stateful SUT \(\mathcal S\) and honest-reference harness
	\(\mathcal H\),
	\[
	\mathsf{Pass}^{\mathsf{href}}_{\Pi^C,\mathcal H}(\mathcal S)
	\le
	\delta_{\Pi^C}
	+
	\epsilon_{\mathsf{alias}}(\mathcal S,\mathcal H)
	+
	\Pr[W\in Q_G(B)]
	+
	\epsilon_{\mathsf{fresh}}(\mathcal S,\mathcal H)
	+
	2^{-\kappa}.
	\]
	For a class \(\mathcal C\), the same bound holds with the class suprema of
	\(\epsilon_{\mathsf{alias}}\), \(\epsilon_{\mathsf{fresh}}\), and the list-hit
	probability.
\end{theorem}

\begin{proof}
	Let \(\Corr=[\Decaps^C(dk,c)=K]\) be the honest correctness event, and let
	\(\Alias_{\mathcal S}\) be as in \Cref{def:alias-freshness}.  Put
	\[
	\Bad=\neg\Corr\cup\Alias_{\mathcal S}.
	\]
	Then
	\[
	\Pr[\Bad]
	\le
	\delta_{\Pi^C}
	+
	\epsilon_{\mathsf{alias}}(\mathcal S,\mathcal H).
	\]
	
	Let \(\Acc=[\widehat K=K]\).  Decompose
	\[
	\Pr[\Acc]
	\le
	\Pr[\Bad]
	+\Pr[\Acc\wedge\neg\Bad\wedge\mathsf{Hit}]
	+\Pr[\Acc\wedge\neg\Bad\wedge\neg\mathsf{Hit}].
	\]
	The middle term is at most \(\Pr[\mathsf{Hit}]=\Pr[W\in Q_G(B)]\).  For the
	final term, condition on
	\(\mathcal E=\neg\Bad\wedge\neg\mathsf{Hit}\) and on a complete output-time
	view \(V=v\).  By the definition of \(\epsilon_{\mathsf{fresh}}\), the
	average posterior guessing probability for \(K\) from that view is at most
	\(2^{-\kappa}+\epsilon_{\mathsf{fresh}}(\mathcal S,\mathcal H)\).  Therefore
	\[
	\Pr[\Acc\wedge\neg\Bad\wedge\neg\mathsf{Hit}]
	\le
	2^{-\kappa}
	+
	\epsilon_{\mathsf{fresh}}(\mathcal S,\mathcal H).
	\]
	Combining the three terms proves the claim.
\end{proof}

\begin{theorem}[Optimality of the hidden-point reduction]
	\label{thm:hidden-point-optimality}
	Let \(\Pi^C\) be a canonical honest-reference line with
	\(\delta_{\Pi^C}=0\), no adversarial aliases, and a secret-separating harness
	in the sense of \Cref{def:alias-freshness}.  Let \(V\) be any SUT view before
	localized final-key queries, and let \(A\) be any measurable
	single-candidate strategy that outputs a suffix \(A(V)\in\mathcal W\).  Then
	there exists a \(1\)-localized SUT \(\mathcal S_A\) such that
	\[
	\mathsf{Pass}^{\mathsf{href}}_{\Pi^C,\mathcal H}(\mathcal S_A)
	=
	\Pr[A(V)=W]
	+
	(1-\Pr[A(V)=W])2^{-\kappa}.
	\]
	Consequently, any black-box honest-reference theorem whose hypotheses do not
	bound the posterior predictability of \(W\) from the SUT view cannot give a
	nontrivial bound below the best single-candidate witness-prediction
	probability, up to the fresh-key coincidence term.
\end{theorem}

\begin{proof}
	The SUT \(\mathcal S_A\) computes \(w=A(V)\), queries
	\(G_K(\langle\mathsf{good},B,w\rangle)\), and outputs the oracle reply.  If
	\(w=W\), this query is the hidden reference point \(x^\star\), and the output
	is exactly \(K=G_K(x^\star)\).  If \(w\neq W\), then by canonical injectivity
	the queried oracle point is distinct from \(x^\star\).  Since the harness is
	secret-separating and no alias occurs, the returned value and \(K\) are
	independent uniform \(\kappa\)-bit strings.  The conditional success
	probability in this case is therefore \(2^{-\kappa}\).  Averaging over the
	two cases yields the equality.  The final statement follows by taking \(A\) to
	be the best single-candidate predictor available from the view.
\end{proof}

\begin{remark}[List-hit tightness]
	\Cref{thm:code-blind} reduces harness acceptance to a list-hit term.
	\Cref{thm:hidden-point-optimality} shows that the term is tight: a predictor
	for the hidden witness induces a one-query SUT with the corresponding pass
	probability, up to fresh-key coincidence.  Source cUP, conditional
	min-entropy, hashed-witness entropy, and TailEnt are the assumptions used to
	bound this term for the different instantiations.
\end{remark}

\begin{remark}[Scope of the list-hit bound]
	\label{rem:list-hit-reading}
	The theorem is a list-hit statement.  Unpredictability of a particular code
	source and coverage of omitted operations are separate hypotheses.  Once
	correctness, adversarial aliasing, and final-key freshness are accounted for,
	honest-reference acceptance is controlled by the SUT's list of localized
	final-key suffix candidates and by final-key guessing.
\end{remark}

\subsection{Harness list-cUP bridge}

The verifiable-decapsulation paper defines confirmation-code unpredictability
(cUP) and faulty implementation correctness (fCOR) for confirmation-code
transforms~\cite{Glabush2025VerifiableDecapsulation}.  In a test harness, a
SUT may make several localized final-key queries; therefore the matching source
notion is a list version of cUP.

\begin{definition}[Harness list-cUP]
	\label{def:list-cup}
	Fix a canonical line \(\Pi^C\), a harness \(\mathcal H\), and a restricted
	oracle interface \(\mathcal F\) matching the cUP interface of the source line.
	The experiment
	\[
	\LcUP^{q}_{\Pi^C,\mathcal H,\mathcal F}(\mathcal A)
	\]
	samples \((ek,dk)\leftarrow\KeyGen(1^\lambda)\) and
	\[
	(K,c,\tau)\leftarrow\Encaps^C(ek),\qquad
	B=B(\tau),\qquad W=W(\tau).
	\]
	The experiment sends \(\mathcal A\) the public and secret data allowed by the
	cUP-style experiment, including the harness data such as \(ek,dk,c\), but not
	\(W\) and not \(K\).  It grants \(\mathcal A\) access to \(\mathcal F\), with
	the localized final-key oracle \(G_K(\langle\mathsf{good},B,\cdot\rangle)\)
	removed or blinded.  The adversary outputs a list
	\(L\subseteq\mathcal W\) with \(|L|\le q\).  The experiment outputs \(1\) iff
	\(W\in L\).  Define
	\[
	\Adv^{\LcUP,q}_{\Pi^C,\mathcal H,\mathcal F}(\mathcal A)
	=
	\Pr[\LcUP^{q}_{\Pi^C,\mathcal H,\mathcal F}(\mathcal A)=1].
	\]
	For a class \(\mathcal C\), take the supremum over
	\(\mathcal A\in\mathcal C\).
\end{definition}

\begin{definition}[cUP-faithful harness certificate]
	\label{def:cup-faithful}
	Fix a canonical augmented line \(\Pi^C\), a source cUP interface
	\(\mathcal F\), a harness \(\mathcal H\), a SUT class \(\mathcal C\), and a
	localization bound \(q\).  A certificate
	\[
	\mathsf{Cert}_{\mathsf{cUP}}
	=
	(\mathsf{Inst},\mathsf{Orc},\mathsf{View},\mathsf{Log})
	\]
	proves that \((\mathcal H,\mathcal F,\mathcal C,q)\) is cUP-faithful if the
	following conditions hold for every \(\mathcal S\in\mathcal C\).
	\begin{enumerate}[label=(C\arabic*),leftmargin=2em]
		\item \textbf{Instance equivalence.} The tuple \((ek,dk,c,B,W)\) sampled by
		\(\mathcal H\) has the same distribution as the tuple sampled in the source
		cUP experiment for \(\Pi^C\).
		\item \textbf{View equivalence.} The initial SUT view produced by
		\(\mathcal H\) is simulatable from the source cUP public and secret inputs
		allowed by \(\mathcal F\), and contains no additional function of \(W\), \(K\),
		or \(x^\star\).
		\item \textbf{Oracle equivalence.} Every oracle reply seen by \(\mathcal S\)
		in \(\mathcal H\), except calls to the localized final-key family
		\(G_K(\langle\mathsf{good},B,\cdot\rangle)\), is exactly simulatable using
		\(\mathcal F\) and independent lazy sampling with the same domains.
		\item \textbf{Final-key logging.} Calls by \(\mathcal S\) to
		\(G_K(\langle\mathsf{good},B,w\rangle)\) are not answered by the list-cUP
		simulator as reference-key evaluations.  They are logged as suffixes \(w\) by
		\(\mathsf{Log}\), and at most \(q\) distinct suffixes are logged in every run.
		\item \textbf{State containment.} The SUT receives no side channel, debug
		output, shared oracle table, deterministic seed, transcript cache, or log
		entry that is absent from the simulated source interface.
		\item \textbf{No extra witness path.} Any computation, leakage, cache, or
		oracle interaction through which \(\mathcal S\) can determine \(W\) must be
		represented in the source cUP interface \(\mathcal F\) and therefore counted
		in the source cUP advantage.  If \(dk\), auxiliary data, or public code give a
		SUT class an alternate way to compute \(W\) that is not simulated by
		\(\mathcal F\), then that class is not cUP-faithful for this harness.
	\end{enumerate}
\end{definition}

\begin{remark}[Scope of the source-cUP route]
	\label{rem:source-cup-scope}
	The source-cUP route is not a statement about arbitrary decapsulation programs
	given \(dk\).  It applies to SUT classes whose complete view and allowed
	computations are simulatable in the source cUP experiment.  An implementation
	that recomputes \(W\), obtains it through an auxiliary channel, or reconstructs
	it by a path outside \(\mathcal F\) is handled by the list-hit theorem itself,
	or by an entropy/diagnostic assumption, not by importing the source cUP bound.
\end{remark}

\begin{lemma}[Certified extraction of the list-cUP adversary]
	\label{lem:certified-lcup-extraction}
	If \((\mathcal H,\mathcal F,\mathcal C,q)\) is cUP-faithful, then every
	\(q\)-localized \(\mathcal S\in\mathcal C\) induces a list-cUP adversary
	\(\mathcal A_{\mathcal S}\) whose output list \(L\) satisfies
	\[
	(B,W,L)\equiv (B,W,Q_G(B))
	\]
	in distribution.
\end{lemma}

\begin{proof}
	The adversary \(\mathcal A_{\mathcal S}\) receives the source cUP instance
	and runs the simulator specified by the certificate.  By instance
	equivalence, the hidden pair \((B,W)\) has the same distribution as in the
	honest-reference harness.  By view and oracle equivalence, every message
	delivered to \(\mathcal S\), except localized final-key answers, has the same
	distribution as in the harness.  When \(\mathcal S\) attempts a localized
	final-key query \(G_K(\langle\mathsf{good},B,w\rangle)\), the adversary logs
	\(w\) and returns a fresh dummy value from an independently sampled table that
	is never used as the reference key.  The resulting execution is the list-cUP
	experiment, with the final-key oracle family removed or blinded.  State
	containment excludes
	additional leakage.  The logged list is therefore precisely \(Q_G(B)\) under
	the simulated run, and its joint distribution with \((B,W)\) is preserved.
\end{proof}

\begin{lemma}[List-cUP to single-cUP loss]
	\label{lem:list-cup-single}
	If the corresponding single-candidate cUP game has advantage at most
	\(\epsilon_{\mathsf{cUP}}\) for every adversary in the same simulation class,
	then
	\[
	\Adv^{\LcUP,q}\le q\epsilon_{\mathsf{cUP}}.
	\]
\end{lemma}

\begin{proof}
	Let \(\mathcal A\) output a list \(L\) with \(|L|\le q\).  A
	single-candidate adversary \(\mathcal B\) runs \(\mathcal A\), and if
	\(L\neq\emptyset\), outputs a uniformly chosen element of \(L\).  Conditioned
	on \(W\in L\), the probability that \(\mathcal B\) outputs \(W\) is at least
	\(1/q\).  Therefore
	\[
	\Pr[\mathcal B\text{ wins cUP}]
	\ge
	\frac{1}{q}\Pr[W\in L].
	\]
	The assumed cUP bound implies \(\Pr[W\in L]\le q\epsilon_{\mathsf{cUP}}\), and
	taking the supremum yields the lemma.
\end{proof}

\begin{theorem}[Certified honest-reference testing from source cUP]
	\label{thm:list-cup}
	\label{thm:certified-lcup}
	Suppose \((\mathcal H,\mathcal F,\mathcal C,q_G)\) admits a cUP-faithful
	harness certificate in the sense of \Cref{def:cup-faithful}.  Then every
	\(q_G\)-localized \(\mathcal S\in\mathcal C\) satisfies
	\[
	\mathsf{Pass}^{\mathsf{href}}_{\Pi^C,\mathcal H}(\mathcal S)
	\le
	\delta_{\Pi^C}
	+
	\epsilon_{\mathsf{alias}}(\mathcal S,\mathcal H)
	+
	\epsilon_{\mathsf{fresh}}(\mathcal S,\mathcal H)
	+
	\Adv^{\LcUP,q_G}_{\Pi^C,\mathcal H,\mathcal F}(\mathcal C)
	+
	2^{-\kappa}.
	\]
	If the source construction provides a single-candidate cUP bound
	\(\Adv^{\mathsf{cUP}}_{\Pi^C,\mathcal F}(\mathcal C)
	\le\epsilon_{\mathsf{cUP}}\), then
	\[
	\mathsf{Pass}^{\mathsf{href}}_{\Pi^C,\mathcal H}(\mathcal S)
	\le
	\delta_{\Pi^C}
	+
	\epsilon_{\mathsf{alias}}(\mathcal S,\mathcal H)
	+
	\epsilon_{\mathsf{fresh}}(\mathcal S,\mathcal H)
	+
	q_G\epsilon_{\mathsf{cUP}}
	+
	2^{-\kappa}.
	\]
\end{theorem}

\begin{proof}
	Apply \Cref{thm:code-blind}.  By
	\Cref{lem:certified-lcup-extraction}, \(\Pr[W\in Q_G(B)]\) in the SUT
	experiment is exactly the list-cUP success probability of the induced
	adversary.  The first bound follows by taking the class supremum.  The second
	bound follows from \Cref{lem:list-cup-single}.
\end{proof}

\begin{proposition}[Honest-reference harness and fCOR]
	\label{prop:fcor-harness}
	When the instance distribution, oracle interface, correctness convention, and
	accepted key are the same as in the source fCOR experiment,
	\[
	\mathsf{Pass}^{\mathsf{href}}_{\Pi^C,\mathcal H}(\mathcal S)
	=
	\Pr[\mathsf{fCOR}^{\Pi^C,\mathcal S}\Rightarrow 1].
	\]
	The list-cUP bridge identifies the additional loss incurred by allowing the
	tested system to try a list of final-key suffixes.
\end{proposition}

\begin{proof}
	Under the stated matching conditions, both experiments sample the same key
	pair, ciphertext, hidden accepted key, and oracle replies, and both accept
	exactly on equality with the honest encapsulated key.  The probabilities are
	therefore identical.
\end{proof}

\subsection{Average conditional min-entropy}

\begin{definition}[Average conditional min-entropy]
	\label{def:avg-minentropy}
	Let \(W\in\mathcal W\) be the hidden confirmation witness and let \(V\) be a
	random variable with respect to which a candidate set \(Q(V)\subseteq\mathcal
	W\) is measurable.  Define
	\[
	P_{\mathsf{guess}}(W\mid V)
	=
	\mathbb E_{v\leftarrow V}
	\left[\max_{w\in\mathcal W}\Pr[W=w\mid V=v]\right].
	\]
	The average conditional min-entropy is
	\[
	\AvgHinf(W\mid V)=-\log_2 P_{\mathsf{guess}}(W\mid V).
	\]
	We say \(W\) has \((\lambda,\epsilon_{\mathsf{sm}})\)-smooth average
	conditional min-entropy given \(V\) if there is an event \(\Good\) with
	\(\Pr[\neg\Good]\le\epsilon_{\mathsf{sm}}\) such that
	\(\AvgHinf(W\mid V,\Good)\ge\lambda\).
\end{definition}

\begin{lemma}[List hit from average min-entropy]
	\label{lem:list-minentropy}
	Let \(Q(V)\subseteq\mathcal W\) be \(V\)-measurable and satisfy
	\(|Q(V)|\le q\).  If \(\AvgHinf(W\mid V)\ge\lambda\), then
	\[
	\Pr[W\in Q(V)]\le q2^{-\lambda}.
	\]
	With \((\lambda,\epsilon_{\mathsf{sm}})\)-smooth min-entropy,
	\[
	\Pr[W\in Q(V)]\le
	\epsilon_{\mathsf{sm}}+q2^{-\lambda}.
	\]
\end{lemma}

\begin{proof}
	For the unsmoothed case,
	\[
	\Pr[W\in Q(V)]
	=
	\sum_v \Pr[V=v]\Pr[W\in Q(v)\mid V=v].
	\]
	For each \(v\),
	\[
	\Pr[W\in Q(v)\mid V=v]
	\le
	|Q(v)|\max_w\Pr[W=w\mid V=v]
	\le
	q\max_w\Pr[W=w\mid V=v].
	\]
	Summing over \(v\) yields \(qP_{\mathsf{guess}}(W\mid V)\le q2^{-\lambda}\).
	The smoothed case decomposes over \(\Good\) and charges
	\(\Pr[\neg\Good]\le\epsilon_{\mathsf{sm}}\).
\end{proof}

\begin{theorem}[Entropic honest-reference detectability]
	\label{cor:entropic}
	In the setting of \Cref{thm:code-blind}, suppose \(\mathcal S\) is
	\(q_G\)-localized and that, for a view \(V\) making \(Q_G(B)\) measurable,
	\[
	\AvgHinf(W\mid V)\ge\lambda.
	\]
	Then
	\[
	\mathsf{Pass}^{\mathsf{href}}_{\Pi^C,\mathcal H}(\mathcal S)
	\le
	\delta_{\Pi^C}
	+
	\epsilon_{\mathsf{alias}}(\mathcal S,\mathcal H)
	+
	\epsilon_{\mathsf{fresh}}(\mathcal S,\mathcal H)
	+
	q_G2^{-\lambda}
	+
	2^{-\kappa}.
	\]
	If only \((\lambda,\epsilon_{\mathsf{sm}})\)-smooth min-entropy holds, then
	\[
	\mathsf{Pass}^{\mathsf{href}}_{\Pi^C,\mathcal H}(\mathcal S)
	\le
	\delta_{\Pi^C}
	+
	\epsilon_{\mathsf{alias}}(\mathcal S,\mathcal H)
	+
	\epsilon_{\mathsf{fresh}}(\mathcal S,\mathcal H)
	+
	\epsilon_{\mathsf{sm}}
	+
	q_G2^{-\lambda}
	+
	2^{-\kappa}.
	\]
\end{theorem}

\begin{proof}
	Combine \Cref{thm:code-blind} with \Cref{lem:list-minentropy}, using
	\(Q(V)=Q_G(B)\).
\end{proof}

\subsection{Hashed diagnostic codes}

\begin{definition}[Hashed diagnostic code]
	\label{def:hashed-diagnostic}
	Let \(R\in\mathcal R\) be a hidden raw witness determined by the honest
	reference transcript.  A diagnostic confirmation code of length \(L\) is
	\[
	W=\mathsf{Trunc}_L(
	\Hcd(\langle\mathsf{cd},B,R\rangle)),
	\]
	where \(\Hcd\) is a domain-separated random oracle independent of \(G_K\).  The
	final accepted key is \(K=G_K(\langle\mathsf{good},B,W\rangle)\).  Let
	\(q_H\) be the number of \(\Hcd\)-queries made by the SUT at points of the
	form \(\langle\mathsf{cd},B,r\rangle\), and let \(q_G\) be the number of
	localized final-key suffix candidates.
\end{definition}

\begin{lemma}[Hashed-witness hit bound]
	\label{lem:hashed-witness}
	Let \(V_0\) be a view with respect to which the SUT strategy before
	\(\Hcd\)-queries is measurable.  If \(\AvgHinf(R\mid V_0)\ge\rho\), then
	\[
	\Pr[W\in Q_G(B)]
	\le
	q_H2^{-\rho}+q_G2^{-L}.
	\]
	With \((\rho,\epsilon_{\mathsf{sm}})\)-smooth min-entropy,
	\[
	\Pr[W\in Q_G(B)]
	\le
	\epsilon_{\mathsf{sm}}+q_H2^{-\rho}+q_G2^{-L}.
	\]
\end{lemma}

\begin{proof}
	Let \(\mathsf{RawHit}\) be the event that the SUT queries
	\(\Hcd(\langle\mathsf{cd},B,R\rangle)\).  By
	\Cref{lem:list-minentropy} applied to the raw witness \(R\),
	\[
	\Pr[\mathsf{RawHit}]\le q_H2^{-\rho}
	\]
	in the unsmoothed case, and gains the smoothing loss in the smooth case.
	Conditioned on \(\neg\mathsf{RawHit}\), the value
	\(\Hcd(\langle\mathsf{cd},B,R\rangle)\) has not been revealed to the SUT, so
	its \(L\)-bit truncation \(W\) is uniform and independent of \(Q_G(B)\).
	Since \(|Q_G(B)|\le q_G\), the conditional hit probability is at most
	\(q_G2^{-L}\).  Combining the cases proves the lemma.
\end{proof}

\begin{theorem}[Hashed diagnostic detectability]
	\label{thm:hashed-diagnostic}
	In the setting of \Cref{def:hashed-diagnostic}, suppose that the raw witness
	has \(\AvgHinf(R\mid V_0)\ge\rho\) for a pre-\(\Hcd\)-query view \(V_0\).
	Then
	\[
	\mathsf{Pass}^{\mathsf{href}}_{\Pi^C,\mathcal H}(\mathcal S)
	\le
	\delta_{\Pi^C}
	+
	\epsilon_{\mathsf{alias}}(\mathcal S,\mathcal H)
	+
	\epsilon_{\mathsf{fresh}}(\mathcal S,\mathcal H)
	+
	q_H2^{-\rho}
	+
	q_G2^{-L}
	+
	2^{-\kappa}.
	\]
	If \(R\) only has \((\rho,\epsilon_{\mathsf{sm}})\)-smooth average conditional
	min-entropy given \(V_0\), add \(\epsilon_{\mathsf{sm}}\) to the right-hand
	side.
\end{theorem}

\begin{proof}
	Combine \Cref{thm:code-blind} with \Cref{lem:hashed-witness}.
\end{proof}

\subsection{Range, tail witnesses, and \hqc{}}

\begin{lemma}[Range does not imply entropy]
	\label{lem:range-no-entropy}
	If a random variable \(C\) has support size at most \(2^L\), then
	\(\AvgHinf(C)\le L\).  However, no positive lower bound on
	\(\AvgHinf(C\mid V)\) follows from the range size alone.
\end{lemma}

\begin{proof}
	Since \(|\operatorname{Supp}(C)|\le 2^L\), some value has probability at least
	\(2^{-L}\), so \(\AvgHinf(C)\le L\).  For the conditional claim, take
	\(V=C\).  Then for every \(v\), the posterior guessing probability is \(1\),
	and \(\AvgHinf(C\mid V)=0\), while the range has not changed.
\end{proof}

\begin{definition}[\hqc{} tail witness in the 2025 salted-SFO line]
	\label{def:hqc-tail}
	In the public \hqc{} KEM specification dated 2025-08-22, encryption
	forms
	\[
	u=r_1+h r_2,\qquad
	v=C.\mathsf{Encode}(m)+\mathsf{Truncate}(s r_2+e,\ell),
	\]
	and the KEM derives
	\[
	(K,\theta)=G(H(ek_{\mathsf{KEM}})\|m\|\mathsf{salt}).
	\]
	The ciphertext is \(c_{\mathsf{KEM}}=(c_{\mathsf{PKE}},\mathsf{salt})\).
	Let
	\[
	\ell_{\mathsf{tail}}=n-n_1n_2
	\]
	and define the tail witness
	\[
	T=(s r_2+e)[n_1n_2,\ldots,n-1]\in\bits^{\ell_{\mathsf{tail}}}.
	\]
	For a retained tail length \(L\le\ell_{\mathsf{tail}}\), let
	\[
	W=\mathsf{Tail}_L(T).
	\]
\end{definition}

\begin{proposition}[\hqc{} retained range budget]
	\label{prop:cdstar}
	For HQC-1, HQC-3, and HQC-5, the retained tail code
	\(W=\mathsf{Tail}_L(T)\) has range size at most \(2^L\) with
	\[
	L\le 5,\qquad L\le 11,\qquad L\le 37,
	\]
	respectively.  A one-byte tail diagnostic has retained range bits
	\[
	\min(8,\ell_{\mathsf{tail}})=5,8,8
	\]
	for HQC-1/3/5.
\end{proposition}

\begin{proof}
	By \Cref{def:hqc-tail}, \(T\in\bits^{\ell_{\mathsf{tail}}}\).  Any retained
	\(L\)-bit projection has support size at most \(2^L\).  The public parameter
	relation is
	\[
	17669-46\cdot384=5,\qquad
	35851-56\cdot640=11,\qquad
	57637-90\cdot640=37.
	\]
\end{proof}

\begin{definition}[\hqc{} tail conditional min-entropy assumption]
	\label{def:tail-unpredictability}
	For a harness class \(\mathfrak H\), SUT class \(\mathcal C\), retained length
	\(L\), and error \(\epsilon_{\mathsf{tail}}\), define
	\[
	\TailEnt_{L,\epsilon_{\mathsf{tail}}}^{\mathfrak H,\mathcal C}
	\]
	to be the assumption that, after sampling a canonical honest-reference
	\hqc{} instance and exposing exactly the harness view \(V\) to every
	\(\mathcal S\in\mathcal C\), excluding \(W\), \(K\), and the private reference
	evaluation \(G_K(\langle\mathsf{good},B,W\rangle)\), one has
	\[
	P_{\mathsf{guess}}(W\mid V)
	\le
	2^{-L}+\epsilon_{\mathsf{tail}}.
	\]
	Equivalently,
	\[
	\AvgHinf(W\mid V)\ge
	-\log_2(2^{-L}+\epsilon_{\mathsf{tail}}).
	\]
\end{definition}

\begin{theorem}[\hqc{} tail-test detectability under TailEnt]
	\label{cor:hqc-law}
	Assume
	\(\TailEnt_{L,\epsilon_{\mathsf{tail}}}^{\mathfrak H,\mathcal C}\).  For any
	\(q_G\)-localized SUT \(\mathcal S\in\mathcal C\),
	\[
	\mathsf{Pass}^{\mathsf{href}}_{\mathsf{HQC}^C,\mathcal H}(\mathcal S)
	\le
	\delta_{\mathsf{HQC}^C}
	+
	\epsilon_{\mathsf{alias}}(\mathcal S,\mathcal H)
	+
	\epsilon_{\mathsf{fresh}}(\mathcal S,\mathcal H)
	+
	q_G(2^{-L}+\epsilon_{\mathsf{tail}})
	+
	2^{-\kappa}.
	\]
\end{theorem}

\begin{proof}
	The TailEnt assumption implies
	\(P_{\mathsf{guess}}(W\mid V)\le2^{-L}+\epsilon_{\mathsf{tail}}\).  The proof
	of \Cref{lem:list-minentropy} yields
	\[
	\Pr[W\in Q_G(B)]
	\le
	q_G(2^{-L}+\epsilon_{\mathsf{tail}}).
	\]
	Apply \Cref{thm:code-blind}.
\end{proof}

\begin{proposition}[Necessity of TailEnt for tail-code detectability]
	\label{prop:tailent-necessary}
	Fix an \hqc{} tail-code harness with zero correctness error, no aliasing, and
	a secret-separating final-key point.  Suppose there is a measurable predictor
	\(A(V)\) for the retained tail code such that
	\[
	\Pr[A(V)=W]=\alpha.
	\]
	Then there exists a \(1\)-localized SUT \(\mathcal S_A\) satisfying
	\[
	\mathsf{Pass}^{\mathsf{href}}_{\mathsf{HQC}^C,\mathcal H}(\mathcal S_A)
	=
	\alpha+(1-\alpha)2^{-\kappa}.
	\]
	Thus any claimed bound below this quantity must rule out such predictors.
	The assumption
	\(\TailEnt_{L,\epsilon_{\mathsf{tail}}}^{\mathfrak H,\mathcal C}\) does
	exactly this by imposing
	\(\alpha\le 2^{-L}+\epsilon_{\mathsf{tail}}\) for every single-candidate
	predictor in the modeled view.
\end{proposition}

\begin{proof}
	The SUT queries \(G_K(\langle\mathsf{good},B,A(V)\rangle)\) and outputs the
	reply.  If \(A(V)=W\), the queried point is the hidden reference point.
	Otherwise it is a distinct canonical oracle point and the returned value
	equals the hidden key only with probability \(2^{-\kappa}\).  Averaging over
	the predictor's success event yields the equality.  The last sentence is the
	single-candidate case of the TailEnt guessing bound.
\end{proof}

\begin{remark}[Reading range-budget rows]
	\label{rem:hqc-range-reading}
	\label{rem:range-budget-rows}
	A retained range of \(L\) bits says only that the code can contain at most
	\(L\) bits of uncertainty.  It implies no lower bound on the posterior
	uncertainty left to a SUT.  The unconditional \hqc{} statement in this paper is
	therefore the range budget in \Cref{prop:cdstar}.  The conditional
	detectability statement is \Cref{cor:hqc-law}, which additionally assumes
	TailEnt.  By \Cref{prop:tailent-necessary}, this split is necessary: if the
	SUT view admits a predictor for the retained tail code, then a one-query SUT
	can pass at that prediction rate up to the final-key coincidence term.
\end{remark}

\subsection{Honest-self, support, and repetition}

\begin{theorem}[Symmetric code erasure passes self-tests and fails reference tests]
	\label{thm:self-vs-reference}
	Let \(\Pi^{-W}\) be the endpoint obtained from \(\Pi^C\) by replacing both
	encapsulation and decapsulation final-key derivation
	\[
	G_K(\langle\mathsf{good},B,W\rangle)
	\]
	with
	\[
	G_K(\langle\mathsf{good},B\rangle).
	\]
	Assume domain separation makes these two oracle inputs distinct for all legal
	\(W\), and assume no oracle aliasing.  Then
	\[
	\mathsf{Pass}^{\mathsf{hself}}_{\Pi^{-W},\mathcal H}(\Pi^{-W})
	\ge
	1-\delta_{\Pi^{-W}},
	\]
	while
	\[
	\mathsf{Pass}^{\mathsf{href}}_{\Pi^C,\mathcal H}(\Pi^{-W})
	\le
	\delta_{\Pi^C}+2^{-\kappa}.
	\]
\end{theorem}

\begin{proof}
	For honest-self testing, both sides use \(\Pi^{-W}\).  Except on the endpoint's
	own correctness failure, encapsulation and decapsulation both use
	\(G_K(\langle\mathsf{good},B\rangle)\).  For honest-reference testing, the
	hidden target is \(G_K(\langle\mathsf{good},B,W\rangle)\), while the SUT
	outputs \(G_K(\langle\mathsf{good},B\rangle)\).  These are distinct oracle
	inputs and the erased-code endpoint never queries the hidden point.  Apply
	\Cref{thm:code-blind} with \(\Pr[\mathsf{Hit}]=0\),
	\(\epsilon_{\mathsf{alias}}=0\), and \(\epsilon_{\mathsf{fresh}}=0\).
\end{proof}

\begin{proposition}[Support equivalence preserves pass probability]
	\label{thm:support-equivalence}
	If \(\mathcal S_0\) and \(\mathcal S_1\) are \(\Delta\)-support-equivalent on
	the honest-reference harness distribution \(\mathcal D_{\mathsf{href}}\), then
	\[
	\left|
	\mathsf{Pass}^{\mathsf{href}}_{\Pi^C,\mathcal H}(\mathcal S_0)
	-
	\mathsf{Pass}^{\mathsf{href}}_{\Pi^C,\mathcal H}(\mathcal S_1)
	\right|
	\le
	\Delta.
	\]
\end{proposition}

\begin{proof}
	Acceptance is an event on the joint distribution of the harness input and SUT
	output.  Total variation distance bounds the difference in probability of any
	event.
\end{proof}

\begin{corollary}[Deterministic support coincidence]
	\label{thm:support-coincidence}
	If two deterministic SUTs output the same value on every input in the support
	of \(\mathcal D_{\mathsf{href}}\), then their honest-reference pass
	probabilities are equal.
\end{corollary}

\begin{corollary}[Decision-only faults on honest support]
	\label{cor:decision-only}
	If a decision-only fault agrees with correct decapsulation on every honestly
	generated reference ciphertext, then it has the same honest-reference pass
	probability as the correct line.  Malformed or differential harnesses sample
	different supports and may expose the fault.
\end{corollary}

\begin{theorem}[Reset amplification]
	\label{rem:repeated}
	In a \(t\)-fold reset-product harness, the instances, oracle namespaces, SUT
	state, and SUT coins of the individual runs are independent after reset.  If
	run \(i\), under its own marginal reset distribution, satisfies
	\(\Pr[\Acc_i]\le p_i\), then
	\[
	\Pr[\Acc_1\wedge\cdots\wedge\Acc_t]
	\le
	\prod_{i=1}^t p_i.
	\]
	If \(p_i\le p\) for all \(i\), the all-pass probability is at most \(p^t\).
	Under the smooth entropic bound with uniform parameters, one may take
	\[
	p=
	\delta_{\Pi^C}
	+
	\epsilon_{\mathsf{alias}}
	+
	\epsilon_{\mathsf{fresh}}
	+
	\epsilon_{\mathsf{sm}}
	+
	q_G2^{-\lambda}
	+
	2^{-\kappa}.
	\]
\end{theorem}

\begin{proof}
	The reset-product harness resets the SUT and uses independent instances with
	independent oracle namespaces, so the run distribution factors.  The product
	bound follows immediately.
\end{proof}

\begin{remark}[Non-reset campaigns]
	For a non-reset stateful campaign,
	\[
	\Pr[\text{all accept}]
	=
	\prod_{i=1}^t
	\Pr[\Acc_i\mid \Acc_1,\ldots,\Acc_{i-1}].
	\]
	Exponentiation is valid only if the per-round hypotheses, including entropy
	and query localization, hold under this conditioning.
\end{remark}

\subsection{Dependency-cone lower bound for honest-reference testing}
\label{sec:dependency-cone-lower-bound}

The preceding theorems show how a confirmation-code harness can prove a
positive statement once the hidden witness is unpredictable or cUP-secure.  We
now prove the matching negative statement: over honest ciphertext support, a
black-box honest-reference harness cannot certify execution of operations that
are outside the dependency cone exposed by the confirmation witness and the
harness observation, relative to implementation classes that contain the
corresponding erasure faults.

\begin{definition}[Intervened trace]
	\label{def:intervened-trace}
	Fix an implementation \(\mathcal I\), an input \(z\), implementation coins
	\(\omega\), and a lazy-sampled oracle table \(\mathcal R\).  Let
	\[
	\Trace_{\mathcal I}(z,\omega,\mathcal R)
	=
	(\mathcal O,\prec,\mathsf{val})
	\]
	be the deterministic trace DAG induced by that run.  For a set
	\(U\subseteq\mathcal O\) and a replacement vector
	\[
	a_U\in\prod_{u\in U}\mathsf{Dom}(u),
	\]
	the intervened trace
	\[
	\Trace_{\mathcal I}(z,\omega,\mathcal R)[U\leftarrow a_U]
	\]
	is obtained by replacing the value of every node in \(U\) by the corresponding
	entry of \(a_U\), then recomputing all descendants in topological order using
	the same local transition functions and the same oracle table
	\(\mathcal R\).  Non-descendants are unchanged.
\end{definition}

\begin{remark}[Trace abstraction]
	\label{rem:trace-abstraction}
	The trace model is a typed, total abstraction of the black-box execution being
	certified.  Control-flow predicates, branch selectors, parsing decisions, and
	error states that can affect the black-box target are represented as trace
	nodes or as components of the target vector.  Dummy values used in an
	intervention must be valid inhabitants of the declared node domains.  A crash,
	undefined behavior, or changed public output is covered only if it appears in
	the declared black-box observation; timing, memory layout, debug output, and
	instrumentation state are outside the theorem unless the harness explicitly
	adds them to the target.
\end{remark}

\begin{definition}[Confirmation-observable target vector]
	\label{def:confirmation-observable-target}
	Let \(\mathcal H_{\mathsf{hon}}\) be an honest-reference harness whose input
	distribution \(\mathcal D_{\mathsf{hon}}\) is supported on honestly generated
	ciphertexts.  The confirmation-observable target vector
	\[
	\Tgt_{\mathsf{cc},\mathcal H}(\Trace)
	\]
	contains exactly the trace values through which the black-box observation of
	\(\mathcal H\) factors on honest support:
	\begin{enumerate}[leftmargin=*]
		\item the localized prefix \(B\), the retained witness \(W\), and the canonical
		final-key point \(x^\star=\langle\mathsf{good},B,W\rangle\);
		\item all final-key-oracle inputs and declared oracle-query inputs that the
		harness is allowed to observe or emulate;
		\item all declared public outputs of the SUT interface; and
		\item the branch selector restricted to honest support, after quotienting out
		support-constant predicates.
	\end{enumerate}
	The harness is \emph{black-box over honest support} if, for any two coupled
	runs with the same value of \(\Tgt_{\mathsf{cc},\mathcal H}\), the complete
	harness observation has the same distribution.  The complete observation is
	written
	\[
	\Obs_{\mathcal H}(\mathcal I;z,\omega,\mathcal R),
	\]
	and includes the visible input/output transcript, the declared oracle query
	log, and the verifier's accept bit, but not undeclared internal trace nodes,
	timing, memory layout, debug output, or instrumentation state.
\end{definition}

\begin{definition}[Support-active dependency cone]
	\label{def:support-active-cone}
	Let \(\mathcal D\) be a distribution over inputs and let \(\Tgt\) be a target
	vector on traces.  A set \(U\subseteq\mathcal O\) is inactive for \(\Tgt\) on
	\(\mathcal D\), written
	\[
	U\perp_{\mathcal D}\Tgt,
	\]
	if for every \(z\in\supp(\mathcal D)\), every implementation coin string
	\(\omega\), every oracle table \(\mathcal R\), and every replacement vector
	\(a_U\),
	\[
	\Tgt\!
	\left(
	\Trace_{\mathcal I}(z,\omega,\mathcal R)[U\leftarrow a_U]
	\right)
	=
	\Tgt\!
	\left(
	\Trace_{\mathcal I}(z,\omega,\mathcal R)
	\right).
	\]
	The support-active dependency cone is the complement of the inactive sets: an
	operation belongs to \(\Cone_{\mathcal D}(\Tgt)\) if changing it can change
	\(\Tgt\) on some point of \(\supp(\mathcal D)\), for some fixing of coins and
	oracle table.  We call
	\[
	\Cone_{\mathsf{cc},\mathcal H}
	=
	\Cone_{\mathcal D_{\mathsf{hon}}}
	(\Tgt_{\mathsf{cc},\mathcal H})
	\]
	the support-active confirmation cone of \(\mathcal H\).
	
	For set-valued claims, inactivity must be certified for the whole set \(U\).
	It is not inferred from singleton non-membership unless the trace semantics
	proves that simultaneous interventions on the set cannot create a new path to
	\(\Tgt\).
\end{definition}

\begin{definition}[Erasure fault and execution certifier]
	\label{def:erasure-and-certifier}
	For \(U\subseteq\mathcal O\), a \(U\)-erasure implementation
	\(\mathcal I^{-U}\) is an implementation that never executes operations in
	\(U\), and instead supplies fixed dummy values at those locations before
	continuing with the remaining transition functions.  Write
	\(\Exec_U(\mathcal I)=1\) if \(\mathcal I\) executes every operation in \(U\)
	on every input in the honest support, and \(\Exec_U(\mathcal I)=0\) if it is a
	\(U\)-erasure implementation.
	
	A class \(\mathcal C\) is \(U\)-erasure closed for a harness
	\(\mathcal H\) if, for every \(\mathcal I\in\mathcal C\) with
	\(\Exec_U(\mathcal I)=1\), at least one coupled \(U\)-erasure implementation
	constructed by \Cref{lem:erasure-coupling} also belongs to \(\mathcal C\).
	Equivalently, \(\mathcal C\) contains the erasure faults against which the
	certifier claims soundness.
	
	A black-box \(n\)-sample certification rule \(\Cert\) for \(U\)-execution sees
	only \(n\) independent observations from \(\Obs_{\mathcal H}\) and outputs
	\(1\) to certify that \(U\) was executed.  It is \((\alpha,\beta)\)-sound and
	complete for \(U\)-execution on a class \(\mathcal C\) if
	\[
	\Pr[\Cert=1\mid \mathcal I]\le \alpha
	\quad\text{for all }\mathcal I\in\mathcal C\text{ with }
	\Exec_U(\mathcal I)=0,
	\]
	and
	\[
	\Pr[\Cert=1\mid \mathcal I]\ge 1-\beta
	\quad\text{for all }\mathcal I\in\mathcal C\text{ with }
	\Exec_U(\mathcal I)=1.
	\]
	A nontrivial certifier has \(\alpha+\beta<1\).
\end{definition}

\begin{lemma}[Erasure coupling outside the cone]
	\label{lem:erasure-coupling}
	Let \(\mathcal H_{\mathsf{hon}}\) be black-box over honest support.  If
	\[
	U\perp_{\mathcal D_{\mathsf{hon}}}\Tgt_{\mathsf{cc},\mathcal H},
	\]
	then for every implementation \(\mathcal I\) there exists a \(U\)-erasure
	implementation \(\mathcal I^{-U}\) such that
	\[
	\Obs_{\mathcal H}(\mathcal I;z,\omega,\mathcal R)
	\equiv
	\Obs_{\mathcal H}(\mathcal I^{-U};z,\omega,\mathcal R)
	\]
	for every \(z\in\supp(\mathcal D_{\mathsf{hon}})\), every coin string
	\(\omega\), and every oracle table \(\mathcal R\).  Consequently the one-run
	observation distributions are identical:
	\[
	\Obs_{\mathcal H}(\mathcal I)
	\equiv
	\Obs_{\mathcal H}(\mathcal I^{-U}).
	\]
\end{lemma}

\begin{proof}
	Construct \(\mathcal I^{-U}\) by replacing every operation in \(U\) by a fixed
	dummy value of the correct syntactic type and then executing the remaining
	code.  For fixed \((z,\omega,\mathcal R)\), the trace of \(\mathcal I^{-U}\)
	is exactly an intervened trace
	\(\Trace_{\mathcal I}(z,\omega,\mathcal R)[U\leftarrow a_U]\) for the dummy
	vector \(a_U\).  By the inactivity hypothesis, the value of
	\(\Tgt_{\mathsf{cc},\mathcal H}\) is unchanged.  Since \(\mathcal H\) is
	black-box over honest support, its complete observation factors through this
	target vector.  Hence the observations are identical pointwise for the coupled
	run.  Averaging over \(z\), coins, and the oracle table proves equality of
	observation distributions.
\end{proof}

\begin{theorem}[Dependency-cone non-certification lower bound]
	\label{thm:dependency-cone-lower-bound}
	Let \(\mathcal H_{\mathsf{hon}}\) be any black-box honest-reference harness
	over honest ciphertext support, and let \(U\subseteq\mathcal O\) be inactive
	for the confirmation-observable target:
	\[
	U\perp_{\mathcal D_{\mathsf{hon}}}\Tgt_{\mathsf{cc},\mathcal H}.
	\]
	Let \(\mathcal C\) be an implementation class that contains an implementation
	\(\mathcal I_1\) with \(\Exec_U(\mathcal I_1)=1\) and a coupled
	\(U\)-erasure implementation \(\mathcal I_0\in\mathcal C\) with
	\(\Exec_U(\mathcal I_0)=0\) and the same black-box transcript distribution.
	Then every \((\alpha,\beta)\)-sound and complete black-box certifier for
	\(U\)-execution over \(\mathcal C\) must satisfy
	\[
	\alpha+\beta\ge 1.
	\]
	In particular, the same conclusion holds for every nonempty
	\(U\)-erasure-closed class \(\mathcal C\) containing an implementation that
	executes \(U\).
	
	More generally, if the erasure construction only achieves
	\[
	\Delta_{\TV}
	\bigl(
	\Obs_{\mathcal H}(\mathcal I),
	\Obs_{\mathcal H}(\mathcal I^{-U})
	\bigr)
	\le \Delta
	\]
	per run for a pair \(\mathcal I,\mathcal I^{-U}\in\mathcal C\) with
	\(\Exec_U(\mathcal I)=1\) and \(\Exec_U(\mathcal I^{-U})=0\), then every
	\(n\)-sample certifier over \(\mathcal C\) must satisfy
	\[
	\alpha+\beta+n\Delta\ge 1.
	\]
\end{theorem}

\begin{proof}
	For the exact statement, the observations of \(\mathcal I_1\) and
	\(\mathcal I_0\) have identical one-run distributions.  Hence their
	\(n\)-sample product distributions are also identical.  Let \(p\) be the
	probability that \(\Cert\) outputs \(1\) under this common distribution.
	Completeness on \(\mathcal I_1\) implies \(p\ge 1-\beta\), while soundness on
	\(\mathcal I_0\) implies \(p\le \alpha\).  Therefore
	\(1-\beta\le \alpha\), equivalently \(\alpha+\beta\ge 1\).  Thus no certifier
	with \(\alpha+\beta<1\) exists.
	
	For the approximate statement, total variation distance between \(n\)-fold
	product distributions is at most \(n\Delta\).  Therefore the certifier's
	output probabilities on the executing implementation \(\mathcal I\) and the
	erasure implementation \(\mathcal I^{-U}\) differ by at most \(n\Delta\).
	Completeness and soundness give
	\[
	1-\beta
	\le
	\Pr[\Cert=1\mid \mathcal I]
	\le
	\Pr[\Cert=1\mid \mathcal I^{-U}] + n\Delta
	\le
	\alpha+n\Delta.
	\]
	Rearranging proves \(\alpha+\beta+n\Delta\ge 1\).
\end{proof}

\begin{corollary}[Cone boundary for code claims]
	\label{cor:cone-boundary}
	An honest-reference confirmation-code test can support an operation-execution
	claim only for operations in the support-active confirmation cone
	\(\Cone_{\mathsf{cc},\mathcal H}\), or for operations exposed by additional
	harness observations, instrumentation, malformed support, or differential
	support.  For operations outside that cone, the strongest correct statement
	over erasure-closed classes, or over any class containing the coupled erasure,
	is a non-certification statement of \Cref{thm:dependency-cone-lower-bound}.
\end{corollary}

\begin{proof}
	If \(U\) is outside the cone, then it is inactive for
	\(\Tgt_{\mathsf{cc},\mathcal H}\), and
	\Cref{thm:dependency-cone-lower-bound} rules out nontrivial black-box
	certification for any implementation class that contains the coupled erasure.
	If a test samples a different support, logs additional internal data, or adds
	a code source that depends on \(U\), then the target vector and its cone have
	changed; the theorem no longer classifies \(U\) as outside the cone.
\end{proof}

\begin{remark}[Indistinguishability from cone inactivity]
	The only structural hypothesis is black-box factorization: the harness
	observation must depend on the implementation trace only through declared
	inputs, declared oracle-query logs, declared outputs, and the final-key test.
	That is the definition of the black-box honest-reference setting.  The theorem
	constructs a \(U\)-erasure with the same transcript distribution whenever
	\(U\) is outside the support-active cone.  The certification lower bound is a
	statement about implementation classes that include that erasure, because
	soundness is always relative to the class on which a certifier is specified.
\end{remark}

\begin{definition}[Cone-inactivity certificate]
	\label{def:cone-certificate}
	A cone-inactivity certificate for an artifact row consists of:
	\begin{enumerate}[leftmargin=*]
		\item the omitted operation set \(U\);
		\item the exact target vector components being checked;
		\item the support condition, such as honest ciphertext support or a named
		malformed/differential support;
		\item a trace argument proving
		\[
		U\perp_{\mathcal D}\Tgt
		\]
		for the whole set \(U\); and
		\item the implementation class over which the coupled erasure is included.
	\end{enumerate}
	Rows without such a certificate remain finite-catalog observations or
	source-code-cone boundaries, depending on the theorem invoked.
\end{definition}

\begin{table}[H]
	\centering
	\footnotesize
	\begin{tabularx}{\textwidth}{@{}>{\raggedright\arraybackslash}p{2.7cm}
			>{\raggedright\arraybackslash}p{3.5cm}
			>{\raggedright\arraybackslash}p{3.7cm}
			>{\raggedright\arraybackslash}X@{}}
		\toprule
		Row family & Omitted set \(U\) & Target components and support & Certificate status \\
		\midrule
		\hqc{}-\cdone{} non-retained tail bits &
		tail coordinates outside the retained \cdone{} projection &
		\(B,W,x^\star\), declared final-key inputs, public outputs, and the
		support-constant branch selector on honest ciphertext support &
		set-level certificate when the black-box target logs only the retained tail
		projection; invalid if a harness exposes the full tail, a full comparison
		trace, or instrumentation state \\
		\mlkem{} unselected source coefficients &
		coefficients outside the selected source set \(S\) &
		source-code target defined by the selected coefficient serialization &
		certificate for the selected source-code cone only; a full
		\(\Tgt_{\mathsf{cc},\mathcal H}\) certificate must additionally exclude
		ordinary comparison outputs, declared query logs, and branch behavior \\
		Always-select-\(K_+\) on honest ciphertexts &
		decision node selecting between \(K_+\) and \(K_-\) &
		honest ciphertext support, after quotienting support-constant predicates &
		certificate only on honest support; malformed or differential support changes
		the target and can make the branch observable \\
		Component-selective comparison rows &
		comparison component not sampled by the row's retained code or mode &
		the named mode and honest-support target for that row &
		row-specific certificate required; singleton non-membership is insufficient
		for a set of comparisons unless simultaneous inactivity is proved \\
		\bottomrule
	\end{tabularx}
	\caption{Cone-inactivity certificates used to interpret negative rows.  A
		row becomes theorem-covered non-certification only for the stated target,
		support, and erasure-containing implementation class.}
	\label{tab:cone-certificates}
\end{table}

\subsection{Coverage supplied by the code source}
\label{sec:coverage}

\begin{definition}[Trace obligations and code coverage]
	\label{def:trace-coverage}
	Let a decapsulation trace be a labeled DAG
	\[
	\Trace=(\mathcal O,\prec,\mathsf{val}),
	\]
	where \(\mathcal O\) is the set of operations, \(\prec\) is the data-dependency
	relation, and \(\mathsf{val}(o)\) is the value produced by operation \(o\).
	Let \(W=\mathsf{Code}(\Trace)\).  For \(U\subseteq\mathcal O\), let
	\(\pi_{\overline U}(\Trace)\) be the trace with all values produced only by
	operations in \(U\) removed.
	
	We say \(U\) is \((\lambda,\epsilon)\)-covered by the code against a SUT class
	\(\mathcal C\) and harness \(\mathcal H\) if, for every
	\(\mathcal S\in\mathcal C\) that omits all operations in \(U\), there is a
	view \(V_U\) containing the SUT transcript such that \(Q_G(B)\) is
	\(V_U\)-measurable and \(W\) has \((\lambda,\epsilon)\)-smooth average
	conditional min-entropy given \(V_U\).
\end{definition}

\begin{theorem}[Covered omitted obligations are detected]
	\label{thm:coverage}
	If \(U\) is \((\lambda,\epsilon)\)-covered by the confirmation code against
	\((\mathcal C,\mathcal H)\), then every \(q_G\)-localized SUT
	\(\mathcal S\in\mathcal C\) omitting \(U\) satisfies
	\[
	\mathsf{Pass}^{\mathsf{href}}_{\Pi^C,\mathcal H}(\mathcal S)
	\le
	\delta_{\Pi^C}
	+
	\epsilon_{\mathsf{alias}}(\mathcal S,\mathcal H)
	+
	\epsilon_{\mathsf{fresh}}(\mathcal S,\mathcal H)
	+
	\epsilon
	+
	q_G2^{-\lambda}
	+
	2^{-\kappa}.
	\]
\end{theorem}

\begin{proof}
	By coverage, the omitted operations leave \(W\) with smooth average conditional
	min-entropy at least \(\lambda\) given the SUT view.  Apply
	\Cref{cor:entropic}.
\end{proof}

\begin{proposition}[Factorization is the one-operation cone case]
	\label{prop:noncoverage}
	Suppose
	\[
	W=f(\pi_{\overline{\{o\}}}(\Trace))
	\]
	for some operation \(o\), and suppose \(o\) also has no active path on honest
	support to any other value in \(\Tgt_{\mathsf{cc},\mathcal H}\).  Then
	\(o\notin\Cone_{\mathsf{cc},\mathcal H}\), and no black-box honest-reference
	certifier over an \(o\)-erasure-closed class can certify execution of \(o\)
	with \(\alpha+\beta<1\).
\end{proposition}

\begin{proof}
	The factorization says that intervening on \(o\) cannot change the witness
	\(W\).  The additional no-active-path condition says that intervening on \(o\)
	cannot change the other components of the confirmation-observable target
	vector either.  Thus
	\[
	\{o\}\perp_{\mathcal D_{\mathsf{hon}}}
	\Tgt_{\mathsf{cc},\mathcal H}.
	\]
	Apply \Cref{thm:dependency-cone-lower-bound}.
\end{proof}

\begin{corollary}[No whole-reencryption conclusion without a cone certificate]
	\label{cor:no-whole-reenc}
	Let \(\mathcal O_{\mathsf{reenc}}\) be the set of operations in the
	reencryption path.  An honest-reference confirmation-code test certifies only
	those operations in
	\[
	\mathcal O_{\mathsf{reenc}}
	\cap
	\Cone_{\mathsf{cc},\mathcal H}
	\]
	for which a positive list-hit, cUP, or entropy bound is supplied.  For every
	operation
	\[
	o\in
	\mathcal O_{\mathsf{reenc}}\setminus\Cone_{\mathsf{cc},\mathcal H},
	\]
	there exists an \(o\)-erasure implementation with identical black-box
	honest-reference transcript distribution on honest support.  Hence no
	nontrivial black-box certificate of whole-reencryption execution follows from
	that harness alone over erasure-closed classes, or over any class containing
	the corresponding erasure implementation.
\end{corollary}

\begin{proof}
	The positive part is \Cref{thm:coverage} or the corresponding cUP/entropy
	bound applied to the operations inside the cone.  The negative part is
	\Cref{thm:dependency-cone-lower-bound} applied to each operation outside the
	support-active confirmation cone.
\end{proof}

\begin{lemma}[\mlkem{} code-source coverage]
	\label{lem:mlkem-coverage}
	In the source \mlkem{} confirmation-code line of
	\cite{Glabush2025VerifiableDecapsulation}, the code is made of selected
	uncompressed ciphertext coefficients from positions
	\(S\subseteq\{1,\ldots,256\}\) in each recomputed \(u_i\) component and in the
	recomputed \(v\) component.  For \(|S|=2\), this selects \(2(k+1)\)
	coefficients, serialized as 12, 16, and 20 bytes for ML-KEM-512, ML-KEM-768,
	and ML-KEM-1024.  Computing that source code through the specified procedure
	requires computing those selected uncompressed reencryption coefficients.
	
	The artifact \mlkem{}-\cdone{} line uses a smaller diagnostic hashed code:
	one byte obtained by hashing 32 residue bytes from selected recomputed
	\(u\)- and \(v\)-component coefficients.  Its formal interpretation is
	\Cref{thm:mlkem-cd1-diagnostic}, a specialization of
	\Cref{thm:hashed-diagnostic} under the named raw-residue entropy assumption.
	It is not covered by the source-paper \mlkem{} cUP theorem.
\end{lemma}

\begin{lemma}[\hqc{} code-source coverage]
	\label{lem:hqc-coverage}
	In the August 2025 \hqc{} \cdone{} and \cdstar{} variants, the confirmation code
	is read from the residual \(v\)-path bits immediately before
	\texttt{vect\_truncate}.  Computing \cdstar{} through the specified procedure
	requires computing the retained tail bits of that pre-truncation \(v\)-buffer;
	computing \cdone{} requires the first \(\min(\ell_{\mathsf{tail}},8)\)
	retained bits.  Full ciphertext comparison and other internal reencryption
	values require separate harness or code-source evidence.
\end{lemma}

\subsection{Source security and ROM/QROM scope}
\label{sec:source-security}

The detectability theorems proved here are classical ROM testing statements for
classical SUTs with classical oracle queries, with construction-level security
of the confirmation-code-augmented KEM line cited from the source
verifiable-decapsulation and FO literature.

In the source notation, cUP is the game for predicting the confirmation code
under restricted access to the internal confirmation function, and fCOR is the
game in which a faulty-but-benign decapsulation outputs the honest encapsulated
key.  The source theorem states a bound of the form
\[
\Adv^{\mathsf{fCOR}}
\le
\delta+\Adv^{\mathsf{cUP}}+2^{-\kappa}
\]
for the confirmation-code-augmented transforms
under the modeled final-key derivation oracle
function~\cite[Theorem~16]{Glabush2025VerifiableDecapsulation}.  Our
contribution is the harness-indexed list version: a concrete test induces a
candidate list \(Q_G(B)\), and the source cUP term controls that list through
\Cref{thm:list-cup} when the simulation interface matches.

The classical query log \(Q_G(B)\) and the event that the hidden point
\(x^\star\) is queried are part of the testing model.  A quantum SUT with
superposition access has no such classical list; a QROM analogue would require
a different theorem, replacing \(\Pr[\Hit]\) by an appropriate query-amplitude
or one-way-to-hiding term for \(x^\star\).  References to QROM results give
construction-level provenance for the underlying augmented FO KEM, not testing
theorems for the harness model~\cite[Theorems~26, 27 and
30]{Glabush2025VerifiableDecapsulation}.

\begin{table}[t]
	\centering
	\footnotesize
	\begin{tabularx}{\textwidth}{@{}>{\raggedright\arraybackslash}p{2.3cm}
			>{\raggedright\arraybackslash}p{4.6cm}
			>{\raggedright\arraybackslash}X@{}}
		\toprule
		Source game or term & Standard role & Use in this paper \\
		\midrule
		cUP & unpredictability of the confirmation code for a faulty implementation with restricted confirmation-function access & single-candidate primitive lifted to list-cUP \\
		list-cUP & probability that a candidate list contains the hidden witness & central hit term for honest-reference testing \\
		fCOR & faulty implementation outputs the honest encapsulated key & matches honest-reference pass probability when interfaces coincide \\
		IND-CCA & ordinary active security of the KEM & inherited from source augmented FO transforms, not reproved here \\
		FFP and spreadness & bad events for simulating decapsulation in FO/QROM proofs & source security provenance \\
		explicit/implicit rejection & invalid-ciphertext behavior of FO-KEM decapsulation & explains why branch-selection faults need their own harness support \\
		\bottomrule
	\end{tabularx}
	\caption{Security vocabulary used in the paper. Construction-level KEM security
		statements are cited from the source games.}
	\label{tab:source-games}
\end{table}

\section{Standard Lines, Augmented Variants, and Mutants}
\label{sec:layers}

The theorem applies to confirmation-code-augmented KEM variants, while the
unmodified standardized KEM, the augmented variant, and the mutant tested
against that variant carry different claims, and the artifact separates these
three objects throughout the experiments.

\begin{table}[t]
	\centering
	\footnotesize
	\begin{tabularx}{\textwidth}{@{}>{\raggedright\arraybackslash}p{3.0cm}
			>{\raggedright\arraybackslash}p{4.4cm}
			>{\raggedright\arraybackslash}X@{}}
		\toprule
		Object & Example in the artifact & Claim made in this paper \\
		\midrule
		Standard line & FIPS 203 \mlkem{} and the public August 2025 \hqc{} line & baseline implementation and source of the decapsulation algorithms \\
		Confirmation-code-augmented variant & \(\VSFO/\cdone{}\), \(\VSFO/\cdstar{}\), and the source-paper \mlkem{} confirmation-code line & object to which the list-hit/list-cUP testing theorem applies \\
		Artifact mutant & drop \(cd\), overwrite \(cd\), always select \(K_+\), compare only one component, symmetric endpoint mutation & tested faulty implementation under a named harness \\
		\bottomrule
	\end{tabularx}
	\caption{Three layers used in the paper.  The augmented variants inherit
		algorithmic structure and formats from the baseline lines, but their
		shared secrets include confirmation-code material and are not output-compatible
		with the unmodified standards.}
	\label{tab:layers}
\end{table}

\begin{definition}[Output compatibility with a standard KEM]
	\label{def:output-compatible}
	An augmented line \(\Pi^C\) is output-compatible with a standard KEM \(\Pi\)
	if, for every standard key pair and every honestly generated standard
	ciphertext, the good-branch shared secret output by \(\Pi^C\) is identical to
	the shared secret output by \(\Pi\), except with the stated correctness error.
\end{definition}

\begin{proposition}[Confirmation-code binding is not output-compatible]
	\label{prop:non-output-compatible}
	Let the standard good key be derived at a domain-separated oracle point
	\(x_{\mathsf{std}}\), and let the augmented good key be derived at
	\(x^\star=\langle\mathsf{good},B,W\rangle\).  If
	\(x_{\mathsf{std}}\neq x^\star\) except with aliasing probability
	\(\epsilon_{\mathsf{alias}}\), then the augmented line is not
	output-compatible with the standard line.  More precisely, on honest
	ciphertexts,
	\[
	\Pr[K_{\mathsf{std}}=K_{\mathsf{aug}}]
	\le
	\epsilon_{\mathsf{alias}}+2^{-\kappa}.
	\]
\end{proposition}

\begin{proof}
	Outside the aliasing event, the standard and augmented keys are outputs of the
	random oracle on distinct domain-separated inputs.  They are therefore
	independent uniform \(\kappa\)-bit strings, and collide with probability
	\(2^{-\kappa}\).  Adding the aliasing event proves the bound.
\end{proof}

\begin{remark}[Deployment reading]
	\label{rem:test-mode-reading}
	The augmented variants in this paper are test-mode or derived research lines.
	They are not drop-in output-compatible modes for FIPS 203 \mlkem{} or for the
	public \hqc{} line.  The standard implementations serve as baselines and
	sources of the decapsulation algorithms; the confirmation-code variants create
	a separate test signal by changing the good-branch key input.
\end{remark}

FIPS 203 defines the standardized \mlkem{} behavior, and the artifact
\(\VSFO/\cdone{}\) line inherits the ciphertext format of \mlkemnative{} while
changing the good-branch shared secret by hashing the recomputed diagnostic
code with the ordinary key material.  It is a test variant over FIPS 203
algorithms, not a FIPS 203-compatible output mode.  The same distinction
applies to the \hqc{} augmented lines: the baseline is the public August 2025
specification, while \cdone{} and \cdstar{} are confirmation-code-augmented
variants used to test the code-binding mechanism.

The standard lines appear in the experiments as baselines for support and
branch behavior, and some representative reencryption omissions are already
visible under honest-reference testing on those baselines.  The positive
theorem concerns the additional localized final-key signal introduced by
binding a hidden confirmation witness into the good key of the augmented
variant.

\section{Instantiations}
\label{sec:instantiations}

\subsection{\mlkem{}}

The \mlkem{} source construction is the confirmation-code line of
\cite{Glabush2025VerifiableDecapsulation}, where the confirmation code is
derived from selected recomputed uncompressed ciphertext coefficients and
included in final key derivation.  If the source cUP bound for the selected
code is \(\varepsilon_{\mathsf{cUP}}^{\mlkem}\), and if the harness and SUT
class satisfy the cUP-faithful certificate of \Cref{def:cup-faithful}, then the
harness list-cUP theorem yields, for a \(q_G\)-localized SUT in that simulated
class,
\[
\mathsf{Pass}^{\mathsf{href}}_{\Pi^C,\mathcal H}(\mathcal S)
\le
\delta_{\Pi^C}
+
\epsilon_{\mathsf{alias}}(\mathcal S,\mathcal H)
+
\epsilon_{\mathsf{fresh}}(\mathcal S,\mathcal H)
+
q_G\varepsilon_{\mathsf{cUP}}^{\mlkem}
+
2^{-\kappa}.
\]

For the selected-code line, the transfer requires both the source cUP theorem
and the cUP-faithful harness certificate.  SUT classes that can compute \(W\)
through paths outside the source cUP interface fall outside this
source-security route.

The source paper proposes short \mlkem{} codes by selecting
\(|S|=2\) coefficients from each recomputed ciphertext component.  With the
usual \mlkem{} module ranks \(k=2,3,4\), this selects 12, 16, and 20 serialized
code bytes for ML-KEM-512, ML-KEM-768, and ML-KEM-1024.  Those source codes are
the faithful \mlkem{} instantiation of the cUP analysis cited above.

\paragraph{Cone accounting for the source \mlkem{} code.}
The dependency-cone theorem assigns an exact interpretation to the source
coefficient selection.  With \(|S|=2\), the source code reads two selected
uncompressed coefficients from each recomputed ciphertext component
\(u_1,\ldots,u_k,v\).  The source-code cone therefore contains \(2(k+1)\)
selected coefficient nodes.  The full recomputed uncompressed \(u\)-and-\(v\)
coefficient array has \(256(k+1)\) coefficient nodes, so \(254(k+1)\)
coefficient nodes are outside this particular selected source-code cone.
Membership outside the full \(\Tgt_{\mathsf{cc},\mathcal H}\) cone requires a
separate target certificate: ordinary comparison outputs, declared query logs,
public outputs, or branch behavior may expose those nodes under a larger
harness target.

\begin{table}[H]
	\centering
	\small
	\begin{tabular}{lrrrr}
		\toprule
		Parameter set & \(k\) & Total \(u/v\) coeffs & Selected coeffs & Outside source-code cone \\
		\midrule
		ML-KEM-512  & 2 & \(256\cdot3=768\)   & \(2\cdot3=6\)  & \(762\) \\
		ML-KEM-768  & 3 & \(256\cdot4=1024\)  & \(2\cdot4=8\)  & \(1016\) \\
		ML-KEM-1024 & 4 & \(256\cdot5=1280\)  & \(2\cdot5=10\) & \(1270\) \\
		\bottomrule
	\end{tabular}
	\caption{Dependency-cone accounting for the source \mlkem{} confirmation code
		with \(|S|=2\).  The selected coefficients are the code-source cone.  The
		outside column is the boundary of this selected
		source-code target.  A full black-box non-certification claim additionally
		requires a cone-inactivity certificate for
		\(\Tgt_{\mathsf{cc},\mathcal H}\), as in \Cref{def:cone-certificate}.}
	\label{tab:mlkem-cone-accounting}
\end{table}

Starting from the pinned \mlkemnative{} v1.0.0 baseline used in the
experiments, the artifact patch materializes a \(\VSFO/\cdone{}\) variant.
This derived line is a one-byte
diagnostic hashed-code variant over the FIPS 203 algorithms, distinct from the
full source-paper short-code construction and from a FIPS 203
output-compatible KEM.  Its proof obligation is therefore
\Cref{thm:mlkem-cd1-diagnostic}, which specializes the hashed-witness theorem
under the named \mlkem{}-\cdone{} raw-residue entropy assumption.  The
source-paper \mlkem{} cUP theorem does not cover this diagnostic line.  The
ciphertext format is inherited from the baseline, while the shared secret is
derived from the recomputed diagnostic code as well as the ordinary \mlkem{}
key material.

The patched encapsulation computes the usual
\[
\mathsf{buf}=m\|H(pk),\qquad
kr=G(\mathsf{buf}),
\]
uses the second 32-byte half of \(kr\) as the IND-CPA coins, and calls
\texttt{mlk\_indcpa\_enc\_cd}.  That routine returns the ciphertext and a
one-byte confirmation code.  The code is the first byte of
\(\mlkem{}\)'s hash \(H\) applied to a 32-byte witness: the first 16 bytes are
compression residues from the first 16 coefficients of the recomputed
\(u\)-component, and the last 16 bytes are compression residues from the first
16 coefficients of the recomputed \(v\)-component.  The artifact good key is
computed as \(H(kr\|cd)\).  Decapsulation decrypts, recomputes the same
ciphertext and code, compares the ciphertext in the ordinary way, computes the
same \(H(kr\|cd)\) value for the good branch, and leaves the rejection branch
unchanged.

\begin{definition}[\mlkem{}-\cdone{} raw-residue entropy assumption]
	\label{def:mlkem-cd1-rawent}
	Let \(R_{\mathsf{ML}}\in\bits^{256}\) be the 32-byte residue witness used by
	the artifact \mlkem{}-\cdone{} line: the first 16 bytes are the selected
	recomputed \(u\)-component compression residues and the last 16 bytes are the
	selected recomputed \(v\)-component compression residues.  Let
	\[
	W=\mathsf{Trunc}_8(\Hcd(\langle\mathsf{cd},B,R_{\mathsf{ML}}\rangle)).
	\]
	For a harness class \(\mathfrak H\), SUT class \(\mathcal C\), and parameters
	\((\rho_{\mathsf{ML}},\epsilon_{\mathsf{ML}})\), the assumption
	\[
	\mathsf{MLKEM\mbox{-}CD1\mbox{-}RawEnt}^{
		\mathfrak H,\mathcal C}_{\rho_{\mathsf{ML}},\epsilon_{\mathsf{ML}}}
	\]
	states that, in the canonical honest-reference \mlkem{}-\cdone{} experiment,
	after exposing the pre-\(\Hcd\)-query SUT view \(V_0\) but excluding
	\(R_{\mathsf{ML}}\), \(W\), \(K\), \(x^\star\), and the private oracle value
	\(\Hcd(\langle\mathsf{cd},B,R_{\mathsf{ML}}\rangle)\), one has
	\[
	P_{\mathsf{guess}}(R_{\mathsf{ML}}\mid V_0)
	\le
	2^{-\rho_{\mathsf{ML}}}
	+
	\epsilon_{\mathsf{ML}}.
	\]
\end{definition}

The RawEnt assumption names a hypothesis about the SUT class and harness view;
without that hypothesis, the artifact rows for \mlkem{}-\cdone{} are
finite-catalog diagnostic observations.

\begin{theorem}[\mlkem{}-\cdone{} diagnostic detectability]
	\label{thm:mlkem-cd1-diagnostic}
	Assume
	\[
	\mathsf{MLKEM\mbox{-}CD1\mbox{-}RawEnt}^{
		\mathfrak H,\mathcal C}_{\rho_{\mathsf{ML}},\epsilon_{\mathsf{ML}}}.
	\]
	For every SUT \(\mathcal S\in\mathcal C\) making at most \(q_H\) raw
	\(\Hcd\)-queries and at most \(q_G\) localized final-key suffix queries,
	\[
	\mathsf{Pass}^{\mathsf{href}}_{\mathsf{MLKEM\mbox{-}CD1},\mathcal H}
	(\mathcal S)
	\le
	\delta_{\mathsf{MLKEM\mbox{-}CD1}}
	+
	\epsilon_{\mathsf{alias}}(\mathcal S,\mathcal H)
	+
	\epsilon_{\mathsf{fresh}}(\mathcal S,\mathcal H)
	+
	q_H(2^{-\rho_{\mathsf{ML}}}+\epsilon_{\mathsf{ML}})
	+
	q_G2^{-8}
	+
	2^{-\kappa}.
	\]
\end{theorem}

\begin{proof}
	The raw-query list has size at most \(q_H\).  The raw-residue entropy
	assumption and the list-hit lemma imply that the probability of querying the
	exact raw witness \(R_{\mathsf{ML}}\) to \(\Hcd\) is at most
	\(q_H(2^{-\rho_{\mathsf{ML}}}+\epsilon_{\mathsf{ML}})\).  Conditioned on not
	making that raw query, the value
	\(\Hcd(\langle\mathsf{cd},B,R_{\mathsf{ML}}\rangle)\) is fresh and its first
	byte is uniform from the SUT view.  A \(q_G\)-localized final-key suffix list
	then hits \(W\) with probability at most \(q_G2^{-8}\).  Apply
	\Cref{thm:code-blind}.
\end{proof}

The \mlkem{} experiments compare plain faults, transformed diagnostic binding
faults, decision-only faults, and symmetric endpoint faults on the one-byte
code-bearing artifact variant named \cdone{}.  That variant is a small,
reproducible harness object for the theorem-relevant fault families, with a
diagnostic security reading separate from the source confirmation-code theorem.

\begin{table}[t]
	\centering
	\footnotesize
	\begin{tabularx}{\textwidth}{@{}>{\raggedright\arraybackslash}p{2.9cm}
			>{\raggedright\arraybackslash}p{4.4cm}
			>{\raggedright\arraybackslash}X@{}}
		\toprule
		\mlkem{} line & Confirmation code & Role in this paper \\
		\midrule
		FIPS 203 \mlkem{} & none & standards baseline and source of the unmodified algorithms \\
		source \(\mlkem{}^C[S]\) & selected uncompressed \(u_i\) and \(v\) coefficients; 12/16/20 bytes for \(|S|=2\) & faithful source-paper instantiation covered by the cited cUP analysis \\
		artifact \(\VSFO/\cdone{}\) & one byte from a hash of 32 recomputed residue bytes & diagnostic non-standard variant governed by \Cref{thm:mlkem-cd1-diagnostic} \\
		\bottomrule
	\end{tabularx}
	\caption{\mlkem{} layers used in the paper.  Only the augmented lines bind a
		confirmation code into the good key; the artifact line uses a smaller code than
		the source-paper short code.}
	\label{tab:mlkem-layers}
\end{table}

\begin{table}[t]
	\centering
	\footnotesize
	\begin{tabularx}{\textwidth}{@{}>{\raggedright\arraybackslash}p{3.0cm}
			>{\raggedright\arraybackslash}p{4.6cm}
			>{\raggedright\arraybackslash}X@{}}
		\toprule
		Object & \mlkem{} reading & Testing role \\
		\midrule
		\(B\) & \(kr\) and the final-key domain material used by the patched good branch & localized prefix for the final key query \\
		\(\Wit\) in \(\mlkem{}^C[S]\) & selected recomputed uncompressed ciphertext coefficients & source confirmation code controlled by the cited cUP theorem \\
		\(\Wit\) in artifact \cdone{} & one-byte hash of the recomputed compression-residue witness & diagnostic hidden code value interpreted through hashed-witness detectability \\
		binding mutant & remove or overwrite the confirmation code before the final key derivation & direct localized-list-hit fault \\
		decision mutant & accept on a branch condition that is unchanged on honest ciphertexts & support-coincidence example \\
		symmetric endpoint mutant & remove the code from both encapsulation and decapsulation & self-vs-reference example \\
		\bottomrule
	\end{tabularx}
	\caption{\mlkem{} interpretation of the abstract quantities.}
	\label{tab:mlkem-map}
\end{table}

\subsection{August 2025 \hqc{}}

For \hqc{}, we use the public specification dated 2025-08-22 as the main
public baseline line~\cite{HQC2025Spec}.  That line uses the salted
implicit-rejection FO transform analyzed for multi-target security in
\cite{GlabushHoevelmannsStebila2025MultiTarget}: the ciphertext is salted, the
encapsulation key is hashed into the derivations, and the derivation transcript
used by the transform is
\[
H(ek_{\mathsf{KEM}})\|m\|\mathsf{salt}.
\]
In our augmented variants, the diagnostic confirmation code comes from the
truncation tail of the \(v\)-path.  The August 2025 salted-SFO line yields a direct
range statement for that tail, formalized in \Cref{prop:cdstar}.  Tail-code
detectability requires \(\TailEnt\), as stated in \Cref{cor:hqc-law}, or
another explicit lower bound for the witness given the SUT view.  We consider
two tail-derived variants:
\begin{itemize}[leftmargin=1.8em]
	\item \cdone{} retains one byte, truncated to the available tail;
	\item \cdstar{} retains the full truncation-tail range budget.
\end{itemize}
The artifact patch computes
\(\theta=G_\theta(H(ek_{\mathsf{KEM}})\|m\|\mathsf{salt})\) with the
\hqc{} domain byte for \(G_\theta\).  The modified PKE encryption routine
\texttt{hqc\_pke\_encrypt\_cd} extracts the confirmation code from the
residual bits of the \(v\)-path immediately before \texttt{vect\_truncate}.
For bit \(i\) of the retained code, the source bit is
\(\mathsf{PARAM\_N1N2}+i\) in the pre-truncation \(v\)-buffer, subject to the
configured bit limit and output byte length.  The good key is computed by
\texttt{hash\_k\_confirmation}, which absorbs
\(H(ek_{\mathsf{KEM}})\), \(m\), \(\mathsf{salt}\), the retained confirmation
code, and the \hqc{} domain byte for \(G_K\).  Decapsulation decrypts,
recomputes \(\theta\), reruns the same code-producing encryption routine,
compares the PKE ciphertext, and uses the same code-bearing good-key hash for
the good branch.

The earlier verifiable-decapsulation discussion of \hqc{} uses its own
confirmation-code analysis and independence assumptions for that setting.  For
the August 2025 salted transcript, we prove the retained range budget and state
the conditional tail-entropy assumption needed to turn that range into a
detectability bound.

\begin{table}[t]
	\centering
	\footnotesize
	\begin{tabularx}{\textwidth}{@{}>{\raggedright\arraybackslash}p{2.4cm}
			>{\raggedright\arraybackslash}p{5.0cm}
			>{\raggedright\arraybackslash}X@{}}
		\toprule
		Step & August 2025 \hqc{} baseline & Augmented \hqc{} variant used here \\
		\midrule
		Encaps randomness and key seed & \((K,\theta)\leftarrow G(H(ek_{\mathsf{KEM}})\|m\|\mathsf{salt})\) & derive the same reencryption transcript, with \(G_\theta\) and \(G_K\) modeled as separated domains \\
		PKE encryption & \(c_{\mathsf{PKE}}\leftarrow \mathsf{HQC\text{-}PKE.Enc}(ek_{\mathsf{KEM}},m;\theta)\) & \((c_{\mathsf{PKE}},cd)\leftarrow \mathsf{HQC\text{-}PKE.EncC}(ek_{\mathsf{KEM}},m;\theta)\) \\
		Good key & \(K\) from the baseline derivation & \(K_+\leftarrow G_K(\langle\mathsf{good},H(ek_{\mathsf{KEM}})\|m\|\mathsf{salt},cd\rangle)\) \\
		Decaps recomputation & recompute \(m'\), \(\theta'\), and \(c'_{\mathsf{PKE}}\), then compare & recompute \(m'\), \(\theta'\), \(c'_{\mathsf{PKE}}\), and \(cd'\), then bind \(cd'\) into \(K_+'\) \\
		Ciphertext format & \((c_{\mathsf{PKE}},\mathsf{salt})\) & unchanged \((c_{\mathsf{PKE}},\mathsf{salt})\) \\
		Shared-secret output & baseline shared secret & changed by confirmation-code binding, hence non-output-compatible with the baseline \\
		\bottomrule
	\end{tabularx}
	\caption{Pseudo-diff for the August 2025 public \hqc{} line and the
		confirmation-code-augmented variants used in the artifact.}
	\label{tab:hqc-pseudodiff}
\end{table}

\paragraph{August 2025 as anchor.}
During 2025, the public \hqc{} line changed in ways that affect the baseline
for the augmented variants.  The February line is useful historical context,
including for the implementation flaw that motivated part of this work, but the
August public specification is the relevant line for the August 2025 \hqc{}
experiments.  The salted transcript, the implicit-rejection flow, and the key
layout are part of the baseline from which the augmented variants are derived.

\begin{table}[t]
	\centering
	\small
	\begin{tabular}{l S[table-format=2.0] S[table-format=2.0] S[table-format=2.0]}
		\toprule
		Parameter set & {\(\ell_{\mathsf{tail}}\)} & {\cdone{} range bits} & {\cdstar{} range bits} \\
		\midrule
		HQC-1 & 5  & 5 & 5 \\
		HQC-3 & 11 & 8 & 11 \\
		HQC-5 & 37 & 8 & 37 \\
		\bottomrule
	\end{tabular}
	\caption{August 2025 \hqc{} tail-derived range budgets.  \(\ell_{\mathsf{tail}}\) is the number
		of residual truncation-tail bits in the code source.  These are range bounds;
		conditional min-entropy is the separate hypothesis used in
		\Cref{cor:hqc-law}.}
	\label{tab:hqc-budgets}
\end{table}

\paragraph{Cone accounting for \hqc{} tail variants.}
For \hqc{}, the retained code cone is exactly the retained truncation-tail
slice.  On the August 2025 salted-SFO line, \(\ell_{\mathsf{tail}}=5,11,37\) for
HQC-1/3/5.  The \cdone{} line retains \(5,8,8\) bits, so it leaves
\(0,3,29\) tail bits outside its retained tail-code cone.  The \cdstar{} line retains the
full tail range and therefore leaves no tail bit outside the tail-code cone,
although operations outside the tail path or outside the honest-support target
remain subject to \Cref{thm:dependency-cone-lower-bound}.

\begin{table}[H]
	\centering
	\small
	\begin{tabular}{lrrrrr}
		\toprule
		Parameter set & \(\ell_{\mathsf{tail}}\) & \cdone{} retained & outside \cdone{} tail cone & \cdstar{} retained & outside \cdstar{} tail cone \\
		\midrule
		HQC-1 & 5  & 5 & 0  & 5  & 0 \\
		HQC-3 & 11 & 8 & 3  & 11 & 0 \\
		HQC-5 & 37 & 8 & 29 & 37 & 0 \\
		\bottomrule
	\end{tabular}
	\caption{Dependency-cone accounting for August 2025 \hqc{} tail-derived codes.
		These are cone and range counts, not entropy claims.  Tail-code detectability
		still requires \TailEnt{} or another unpredictability bound for the retained
		bits, while non-retained tail bits give theorem-covered non-certification only
		when the row includes a cone-inactivity certificate for the black-box target,
		as in \Cref{tab:cone-certificates}.}
	\label{tab:hqc-cone-accounting}
\end{table}

\begin{table}[t]
	\centering
	\footnotesize
	\begin{tabularx}{\textwidth}{@{}>{\raggedright\arraybackslash}p{2.4cm}
			>{\raggedright\arraybackslash}p{4.2cm}
			>{\raggedright\arraybackslash}X@{}}
		\toprule
		HQC aspect & February 2025 line & August 2025 line used here \\
		\midrule
		Transform shape & earlier salted/HHK-style presentation & public salted implicit-rejection flow with explicit \(H(ek_{\mathsf{KEM}})\|m\|\mathsf{salt}\) transcript \\
		Public-key binding & narrower historical binding surface & all bytes of \(ek_{\mathsf{KEM}}\) enter the key and randomness derivation through the public hash \\
		Secret-key layout & historical appended layout & August 2025 public layout and compressed alternatives \\
		Code source & truncation-tail opportunity present & truncation-tail opportunity still present and parameterized \\
		Role in artifact & drift comparison & main public baseline \\
		\bottomrule
	\end{tabularx}
	\caption{Condensed view of the \hqc{} drift relevant to this paper.}
	\label{tab:hqc-drift-main}
\end{table}

\begin{table}[t]
	\centering
	\footnotesize
	\begin{tabularx}{\textwidth}{@{}>{\raggedright\arraybackslash}p{3.0cm}
			>{\raggedright\arraybackslash}p{4.6cm}
			>{\raggedright\arraybackslash}X@{}}
		\toprule
		Object & August 2025 \hqc{} reading & Testing role \\
		\midrule
		\(B\) & \(H(ek_{\mathsf{KEM}})\|m\|\mathsf{salt}\) & localized final-key prefix fixed by the salted reference trace \\
		\(\theta\) & \(G_\theta(B)\) & reencryption randomness that must be recomputed \\
		\(\Wit\) & retained truncation-tail code \(cd\) & hidden code entering the final good-key query \\
		\cdone{} & one-byte retained code, truncated to the available tail on HQC-1 & small August 2025-line variant with a range budget; detectability requires TailEnt or another entropy claim \\
		\cdstar{} & full retained truncation-tail code & full-range variant for the same family; range is not unpredictability \\
		\bottomrule
	\end{tabularx}
	\caption{August 2025 \hqc{} interpretation of the abstract quantities.}
	\label{tab:hqc-map}
\end{table}

\subsection{Instantiation Contributions}

\mlkem{} links the testing statement to a source confirmation-code construction
derived from a standardized KEM and also supplies a small diagnostic
hashed-code artifact line.  The August 2025 \hqc{} line has a public salted
flow with a tail-code family whose retained range can be swept experimentally.
The main theorem applies to augmented variants once the list-hit term is
bounded by source cUP, by min-entropy, or by an explicit diagnostic assumption.
The standard lines provide baselines and implementation context, and the
\hqc{} \cdone{} versus \cdstar{} comparison varies retained range and measured
witness-guess behavior while keeping the KEM design goal fixed.

\begin{table}[t]
	\centering
	\footnotesize
	\begin{tabularx}{\textwidth}{@{}>{\raggedright\arraybackslash}p{2.4cm}
			>{\raggedright\arraybackslash}p{4.5cm}
			>{\raggedright\arraybackslash}X@{}}
		\toprule
		Family & Hidden code value & Main role in this paper \\
		\midrule
		\mlkem{} & source confirmation code from the verifiable-decapsulation line; one-byte hashed diagnostic code in the artifact & source-paper instantiation plus diagnostic artifact variant over FIPS 203 algorithms \\
		August 2025 \hqc{} & truncation-tail diagnostic code in the salted public line & salted public baseline plus range-budget diagnostic variants under explicit TailEnt assumptions when used for detectability \\
		\bottomrule
	\end{tabularx}
	\caption{The two concrete instantiations used in the paper.}
	\label{tab:instantiations}
\end{table}

\begin{table}[t]
	\centering
	\footnotesize
	\begin{tabularx}{\textwidth}{@{}>{\raggedright\arraybackslash}p{3.0cm}
			>{\raggedright\arraybackslash}p{2.5cm}
			>{\raggedright\arraybackslash}X@{}}
		\toprule
		Artifact line & Final-key length \(\kappa\) & Reading \\
		\midrule
		\mlkem{}-\cdone{} & 256 bits & SHA3/SHAKE-derived 32-byte shared secret in the diagnostic augmented line \\
		\hqc{}-\cdone{} & 256 bits & 32-byte shared secret in the augmented salted line \\
		\hqc{}-\cdstar{} & 256 bits & same final-key length, with a larger retained code input \\
		\bottomrule
	\end{tabularx}
	\caption{Final-key length used for the \(2^{-\kappa}\) term in the artifact
		variants.  The term is a random-oracle-model final-key coincidence term.}
	\label{tab:kappa}
\end{table}

\section{Artifact}
\label{sec:artifact}

FO-FaultBench packages the experiment code, patch lines, mutant catalogs,
frozen JSONL records, three derived summaries, and two plots for the confirmation-code
decapsulation experiments.  The package separates a read-only validation path
from a rebuild path: \texttt{make check} validates the stored layout and parses
the JSONL records, \texttt{make summary} prints a compact inventory, and the
fetch/build/run targets acquire pinned sources and rerun the campaigns.

\paragraph{Baselines.}
The main baselines are \mlkemnative{} v1.0.0 and the August 2025 public \hqc{}
reference line.  The \hqc{} baseline is built with an explicit portability
patch on the development host.  The derived lines are materialized from local
patches: \path{mlkem-vsfo-cd1}, \path{hqc-vsfo-cd1}, and
\path{hqc-vsfo-cdstar}.  The historical February 2025 \hqc{} line is retained
for standards-drift comparison.

\paragraph{Harnesses.}
The artifact implements \(\Dhonref\), \(\Dhonself\), \(\Dctmal\),
\(\Ddiff\), and \(\Dkeyfmt\), using the same harness names as the theorem
statements; every ``detected'' or ``undetected'' row is relative to the stated
harness.

\paragraph{Mutants.}
The mutant catalogs contain deterministic recomputation and binding faults,
decision-only faults, symmetric endpoint faults, probabilistic code-guessing
faults, and the \hqc{} range-budget sweep.  Each catalog entry records the
fault class, affected operation, expected harness behavior, and patch.

\paragraph{Reporting.}
Deterministic tables report pass/fail counts over the frozen seeds and
parameter sets.  Probabilistic rows report observed pass
probability and a 95\% interval or upper bound.  For a fixed
line--fault--harness row with \(x\) passes in \(n\) seeded harness runs, the
point estimate is \(\hat p=x/n\).  Nonzero rows use the Wilson score interval
with \(z_{0.975}=1.9599639845\),
\[
\frac{\hat p+z^2/(2n)\pm z\sqrt{\hat p(1-\hat p)/n+z^2/(4n^2)}}{1+z^2/n}.
\]
Zero-pass rows use the upper endpoint \(1-(0.025)^{1/n}\) of the two-sided
95\% Clopper--Pearson interval.  Repeated-test survival columns are computed
as \(\hat p^t\) under the reset-product interpretation of
\Cref{rem:repeated}, using the zero-pass upper endpoint where applicable.

\paragraph{Reproducibility.}
The packaged snapshot contains eleven raw JSONL files, three Markdown
summaries, two SVG plots, and a hash file over the stored package inputs. The
read-only path uses \texttt{make check} and \texttt{make summary}; it validates the
expected directories and files, parses every stored JSONL record, and reports
the inventory without rewriting the stored results. The current snapshot
contains 3592 raw records. The rebuild path uses \texttt{make fetch},
\texttt{make verify-upstreams}, \texttt{make build-baselines},
\texttt{make prepare-variants}, \texttt{make run-core}, and
\texttt{make run-mutants}; it fetches pinned upstreams, verifies them against
\texttt{artifact/baselines.json}, builds the baseline and derived
implementations, and reruns the harness families. Reruns may create build
products and append new JSONL records under \texttt{results/raw/}.

Each source-security row in the artifact records the interface classification
used to interpret the corresponding harness result.  Rows using the source
\mlkem{} confirmation code are eligible for cUP-faithful certification, whereas
rows using \cdone{} or \cdstar{} diagnostic codes are interpreted through the
diagnostic or TailEnt route.

\begin{table}[H]
	\centering
	\footnotesize
	\begin{tabularx}{\textwidth}{@{}>{\raggedright\arraybackslash}p{3.0cm}
			>{\raggedright\arraybackslash}p{3.3cm}
			>{\raggedright\arraybackslash}X@{}}
		\toprule
		Campaign & Main harnesses & Stored records \\
		\midrule
		deterministic recomputation and binding & \(\Dhonref\), \(\Dctmal\), \(\Ddiff\) & raw mutant and differential JSONL records, with the compact mutant summary \\
		decision-only and comparison faults & \(\Dctmal^\mu\), \(\Ddiff^\mu\), \(\Dhonref\) & raw mutant and differential JSONL records \\
		symmetric endpoint faults & \(\Dhonref\), \(\Dhonself\) & raw honest-reference, honest-self, and mutant JSONL records \\
		probabilistic code guessing & \(\Dhonref\) & raw probabilistic-witness JSONL records and compact probability summary \\
		\hqc{} range-budget sweep & \(\Dhonref\) & raw sweep JSONL records, compact sweep summary, and frozen plot \\
		\hqc{} standards drift & baseline/spec comparison & pinned baseline metadata and raw comparison records \\
		\bottomrule
	\end{tabularx}
	\caption{Main artifact campaigns.}
	\label{tab:campaigns}
\end{table}

\begin{table}[H]
	\centering
	\footnotesize
	\begin{tabularx}{\textwidth}{@{}>{\raggedright\arraybackslash}p{2.6cm}
			>{\raggedright\arraybackslash}p{3.0cm}
			>{\raggedright\arraybackslash}X@{}}
		\toprule
		Line & Status in package & Role \\
		\midrule
		\path{mlkem-native-v1.0.0} & pinned, built, smoke-validated & plain \mlkem{} baseline \\
		\path{mlkem-vsfo-cd1} & derived by patch and validated & transformed \mlkem{} code line \\
		\path{hqc-current-ref} & pinned August 2025 public line with portability patch & August 2025 \hqc{} baseline \\
		\path{hqc-vsfo-cd1} & derived by patch and validated & August 2025 \hqc{} one-byte code line \\
		\path{hqc-vsfo-cdstar} & same derived line with full code retention & August 2025 \hqc{} full-budget code line \\
		\texttt{\seqsplit{hqc-historical-feb-2025-additional-128}} & acquired for comparison & standards-drift context \\
		\bottomrule
	\end{tabularx}
	\caption{Baseline and derived-line inventory used by the experiments.}
	\label{tab:baseline-inventory}
\end{table}

\begin{definition}[Artifact ROM-interface certificate]
	\label{def:artifact-rom-certificate}
	An artifact row has a ROM-interface certificate if the row records the
	following structural correspondence between the concrete harness and the
	oracle model.
	\begin{enumerate}[label=(R\arabic*),leftmargin=2em]
		\item Each concrete hash, XOF, or KDF call used as a modeled oracle call is
		assigned to a length-delimited, domain-separated namespace.  Calls modeled as
		\(G_K\), \(\Hcd\), or another oracle either use distinct namespaces or use an
		injective domain tag inside a single concrete primitive.
		\item The final-key point \(x^\star=\langle\mathsf{good},B,W\rangle\) and the
		projection \(Q_G(B)\) use the same serialization as the concrete harness.
		\item All SUT calls to the modeled final-key namespace are logged before
		projection to \(Q_G(B)\), and the SUT has no access to a shared oracle table or
		unlogged cache for \(G_K(x^\star)\).
		\item The SUT input excludes \(W\), \(K\), \(x^\star\), the private reference
		evaluation \(G_K(x^\star)\), deterministic witness-sampling seeds, and debug
		or trace data containing witness bytes.
		\item Any disabled item, noncanonical serialization, shared cache, or
		implementation leak affecting the hidden witness or final key is reported as
		an explicit aliasing or freshness defect for that row.
	\end{enumerate}
\end{definition}

\begin{proposition}[Structural artifact-to-ROM correspondence]
	\label{prop:artifact-rom-correspondence}
	If an artifact row satisfies \Cref{def:artifact-rom-certificate}, then in the
	ROM interpretation induced by that certificate the logged localized final-key
	queries are exactly the modeled list \(Q_G(B)\), and the row has
	\(\epsilon_{\mathsf{fresh}}=0\) except for explicitly reported defects.  The
	claim is a structural correspondence between the artifact interface and the
	ROM model; random-oracle ideality for the concrete SHA3/SHAKE primitives is
	the usual modeling assumption, not an empirical conclusion of the artifact.
\end{proposition}

\begin{proof}
	By (R1) and (R2), concrete final-key calls have a unique modeled namespace and
	canonical projection to \(B\)-localized candidates.  By (R3), every SUT call to
	that namespace appears in the logged list, and no shared table or cache exposes
	the SUT the private reference value without a logged query.  By (R4), the SUT
	view omits the hidden witness, final key, final-key point, and witness-bearing
	sampling or debug data.  Thus the row has no unreported path that increases the
	posterior guessing probability of \(K=G_K(x^\star)\) after conditioning on no
	list hit.  Any violation is charged by (R5) as aliasing or freshness defect.
\end{proof}

\section{Experimental Results}
\label{sec:results}

\begin{definition}[Artifact claim levels]
	\label{def:artifact-claim-levels}
	An experimental row has one of three claim levels.
	\begin{enumerate}[label=(A\arabic*),leftmargin=2em]
		\item \textbf{Finite-catalog claim.} The row concerns exactly the listed
		mutants, seeds, parameter sets, and harnesses.
		\item \textbf{Theorem-covered class claim.} The row is paired with a formal
		class hypothesis such as cUP-faithfulness, conditional min-entropy,
		raw-witness entropy, TailEnt, support equivalence, or code-source coverage.
		\item \textbf{No class claim.} If neither (A1) nor (A2) is stated, the row is
		not evidence for unenumerated faults.
	\end{enumerate}
\end{definition}

\begin{proposition}[Finite-catalog soundness of artifact rows]
	\label{prop:finite-catalog-soundness}
	Let \(\mathcal M\) be the finite mutant catalog named by a table row, let
	\(\mathcal D_n\) be the frozen seeded harness distribution used by the
	artifact, and let \(x_m\) be the number of accepting runs for
	mutant \(m\in\mathcal M\) over \(n_m\) runs.  The unconditional empirical
	statement of the row is exactly the collection of estimates
	\(\widehat p_m=x_m/n_m\) with the stated confidence-interval convention.  It
	implies no statement about a mutant outside \(\mathcal M\) unless an
	additional theorem-covered class claim is invoked.
\end{proposition}

\begin{proof}
	The replay script fixes the baseline, patch, mutant, parameter set, harness,
	seed schedule, and acceptance predicate.  Therefore each \(x_m/n_m\) is an
	empirical estimate for that named finite experiment.  For any mutant
	\(m'\notin\mathcal M\), one can define a deterministic implementation that
	agrees with a tested mutant on all frozen inputs and differs on an untested
	input.  The finite transcript distribution cannot distinguish these two
	implementations.  Hence the finite row alone cannot imply any statement about
	\(m'\).
\end{proof}

The experimental section reports rows as empirical witnesses for
theorem-relevant behavior under named harnesses and fixed catalogs.  After
\Cref{thm:dependency-cone-lower-bound}, a negative honest-reference row has a
formal interpretation when the omitted operation is outside the support-active
confirmation cone and the implementation class contains the coupled erasure:
no black-box honest-support certifier can certify that operation without
changing the code source, adding instrumentation, or sampling a different
support.  A positive row is a finite-catalog observation unless it is linked to
one of the theorem-covered routes: source cUP, conditional entropy,
hashed-witness entropy, TailEnt, support equivalence, self/reference
separation, or dependency-cone non-certification.  Negative rows use the last
route only when a cone-inactivity certificate of \Cref{def:cone-certificate} is
supplied.

\begin{table}[H]
	\centering
	\footnotesize
	\begin{tabularx}{\textwidth}{@{}>{\raggedright\arraybackslash}p{3.0cm}
			>{\raggedright\arraybackslash}p{3.1cm}
			>{\raggedright\arraybackslash}p{3.2cm}
			>{\raggedright\arraybackslash}X@{}}
		\toprule
		Artifact family & Honest-reference claim level & Additional harness claim & Reason \\
		\midrule
		Source \mlkem{} selected-code faults & eligible for theorem-covered positive status once source cUP and a cUP-faithful harness certificate are supplied & none needed for selected code nodes beyond that certificate & selected coefficients lie in the source confirmation cone \\
		\mlkem{} unselected coefficient omissions & selected-source-cone boundary; full non-certification only with a \(\Tgt_{\mathsf{cc},\mathcal H}\) cone certificate & requires larger selected set, instrumentation, malformed support, or differential support & outside the selected source-code cone by \Cref{tab:mlkem-cone-accounting}; see \Cref{tab:cone-certificates} for the extra full-target condition \\
		\mlkem{}-\cdone{} drop/overwrite diagnostic code & theorem-covered positive under hashed-witness theorem plus RawEnt, or finite-catalog if RawEnt is not asserted & none for the named deterministic mutants & hidden diagnostic code changes the final-key point \\
		Plain reencryption omissions & finite-catalog unless a separate ordinary-path cone certificate is supplied & malformed/differential rows may be theorem-covered for the sampled mode & may be visible through ordinary output behavior, not through the confirmation-code cone alone \\
		\hqc{}-\cdone{} retained tail-bit faults & theorem-covered positive under \TailEnt{} for retained bits & none for retained tail cone & retained bits enter the hidden final-key point \\
		\hqc{}-\cdone{} non-retained tail-bit omissions & theorem-covered non-certification only with the retained-tail cone certificate & use \cdstar{}, instrumentation, or a harness exposing those bits & outside \cdone{} retained tail cone by \Cref{tab:hqc-cone-accounting,tab:cone-certificates} \\
		\hqc{}-\cdstar{} retained tail-bit faults & theorem-covered positive under \TailEnt{} for the full tail & none for tail bits & all tail bits enter the retained code cone \\
		Always-select-\(K_+\) decision faults on honest ciphertexts & theorem-covered non-certification only with the honest-support branch certificate & malformed/differential support establishes the positive branch claim & branch selector is support-constant on honest ciphertexts; see \Cref{tab:cone-certificates} \\
		Component-selective comparison faults on honest ciphertexts & theorem-covered non-certification only for rows with a mode-specific cone certificate & differential support establishes the mode-selective claim & comparison target must be certified inactive as a set for that mode \\
		Symmetric endpoint faults & theorem-covered separation & honest-self and honest-reference must both be reported & peer model changes the sampled reference key \\
		Witness-guess mutants & theorem-covered quantitative row for the named guessing interface & reset-product theorem for repeated independent trials & pass probability follows the retained-code list-hit scale \\
		\bottomrule
	\end{tabularx}
	\caption{Artifact claim levels after the dependency-cone theorem.  Negative
		honest-reference rows are upgraded from mere empirical absences to
		theorem-covered non-certification statements only when the omitted operation
		set has a cone-inactivity certificate for the stated target, support, and
		erasure-containing implementation class.}
	\label{tab:artifact-cone-claim-levels}
\end{table}

\subsection{Binding faults}

\begin{table}[H]
	\centering
	\footnotesize
	\begin{tabularx}{\textwidth}{@{}>{\raggedright\arraybackslash}p{2.0cm}
			>{\raggedright\arraybackslash}p{2.3cm}
			>{\raggedright\arraybackslash}p{2.5cm}
			cc
			>{\raggedright\arraybackslash}X@{}}
		\toprule
		Line & Mutant family & Fault class & \(\Dhonref\) & \(\Dctmal\) & Reading \\
		\midrule
		\mlkem{} plain & reencryption omission & recomputation & detected (3/3) & undetected (0/3) & the listed recomputation mutants are visible on this honest support \\
		\mlkem{} VSFO/\cdone{} & overwrite or drop \(cd\) & binding & detected (6/6) & undetected (0/6) & the transformed final key exposes the missing code \\
		\mlkem{} VSFO/\cdone{} & always select \(K_{+}\) & decision-only & undetected (0/3) & detected (3/3) & malformed support samples this branch behavior \\
		August 2025 \hqc{} plain & reencryption omission & recomputation & detected (3/3) & undetected (0/3) & the August 2025 public line exposes this listed recomputation mutant on this support \\
		August 2025 \hqc{} VSFO/\cdone{} & overwrite or drop \(cd\) & binding & detected (6/6) & undetected (0/6) & the salted line exposes the missing code \\
		August 2025 \hqc{} VSFO/\cdone{} & always select \(K_{+}\) & decision-only & undetected (0/3) & detected (3/3) & malformed support samples this branch behavior \\
		\bottomrule
	\end{tabularx}
	\caption{Main deterministic mutant families over the frozen parameter-set
		tranche.  Counts are detected families over tested rows for the stated harness.}
	\label{tab:binding-results}
\end{table}

\begin{observation}
	On both scheme families, the listed plain reencryption-omission mutants are
	detected under honest-reference testing on the frozen support.  The transformed rows isolate
	the additional signal introduced by the confirmation-code path: faults that
	remove, overwrite, or fail to bind the recomputed code change the good-key
	input under \(\Dhonref\), whereas branch-only faults require malformed or
	differential support.
\end{observation}

\subsection{\cdone{} versus \cdstar{}}

\begin{table}[H]
	\centering
	\footnotesize
	\begin{tabularx}{\textwidth}{@{}>{\raggedright\arraybackslash}p{3.2cm}
			>{\raggedright\arraybackslash}p{2.9cm}
			>{\raggedright\arraybackslash}p{2.9cm}
			>{\raggedright\arraybackslash}X@{}}
		\toprule
		Fault family on August 2025 \hqc{} & \cdone{} & \cdstar{} & Reading \\
		\midrule
		overwrite \(cd\) & detected (3/3) under \(\Dhonref\) & detected (3/3) under \(\Dhonref\) & both variants bind the code into the final good key \\
		drop \(cd\) from \(K_{+}\) & detected (3/3) under \(\Dhonref\) & detected (3/3) under \(\Dhonref\) & deterministic behavior is unchanged \\
		always select \(K_{+}\) & undetected (0/3) under \(\Dhonref\) & undetected (0/3) under \(\Dhonref\) & the support limitation is unchanged \\
		compare only \(u\) & mode-selective (0/3 on honest support) & mode-selective (0/3 on honest support) & both variants keep the same comparison surface \\
		\bottomrule
	\end{tabularx}
	\caption{August 2025 \hqc{} \cdone{} and \cdstar{} on the deterministic mutant
		tranche.  The variants differ quantitatively; the deterministic fault classes
		coincide.}
	\label{tab:cdstar-det}
\end{table}

\subsection{Self-vs-reference tests}

\begin{table}[H]
	\centering
	\small
	\begin{tabularx}{\textwidth}{@{}>{\raggedright\arraybackslash}p{3.0cm}
			>{\raggedright\arraybackslash}p{2.3cm}
			cc
			>{\raggedright\arraybackslash}X@{}}
		\toprule
		Line & Symmetric mutant & \(\Dhonref\) & \(\Dhonself\) & Reading \\
		\midrule
		\mlkem{} VSFO/\cdone{} & drop \(cd\) at both endpoints & detected (3/3) & undetected (0/3) & self-tests pass because both endpoints agree on the same faulty key input \\
		August 2025 \hqc{} VSFO/\cdone{} & drop \(cd\) at both endpoints & detected (3/3) & undetected (0/3) & the salted line has the same peer-model separation \\
		August 2025 \hqc{} VSFO/\cdstar{} & drop \(cd_\star\) at both endpoints & detected (3/3) & undetected (0/3) & the separation is independent of the retained code budget \\
		\bottomrule
	\end{tabularx}
	\caption{Self-consistent endpoint faults.  These rows instantiate
		\Cref{thm:self-vs-reference}.}
	\label{tab:self-reference}
\end{table}

\subsection{Probabilistic code guessing}

\begin{table}[H]
	\centering
	\scriptsize
	\begin{tabular}{l S[table-format=2.0] S[table-format=4.0]
			S[table-format=3.0] S[table-format=1.4] l l}
		\toprule
		Line & {Bits} & {Trials} & {Passes} & {Est.\ pass prob.} & {95\% CI / ub} & {Ref.\ scale} \\
		\midrule
		ML-KEM-512-\cdone{} & 8 & 4096 & 14 & 0.0034 & {[0.0020, 0.0057]} & \(2^{-8}\) \\
		ML-KEM-768-\cdone{} & 8 & 4096 & 10 & 0.0024 & {[0.0013, 0.0045]} & \(2^{-8}\) \\
		ML-KEM-1024-\cdone{} & 8 & 4096 & 24 & 0.0059 & {[0.0039, 0.0087]} & \(2^{-8}\) \\
		HQC-1-\cdone{} & 5 & 4096 & 151 & 0.0369 & {[0.0315, 0.0431]} & \(2^{-5}\) \\
		HQC-3-\cdone{} & 8 & 4096 & 21 & 0.0051 & {[0.0034, 0.0078]} & \(2^{-8}\) \\
		HQC-5-\cdone{} & 8 & 4096 & 13 & 0.0032 & {[0.0019, 0.0054]} & \(2^{-8}\) \\
		HQC-1-\cdstar{} & 5 & 4096 & 151 & 0.0369 & {[0.0315, 0.0431]} & \(2^{-5}\) \\
		HQC-3-\cdstar{} & 11 & 8192 & 5 & 0.0006 & {[0.0003, 0.0014]} & \(2^{-11}\) \\
		HQC-5-\cdstar{} & 37 & 4096 & 0 & 0 & {[0, 0.0009]} & \(2^{-37}\) \\
		\bottomrule
	\end{tabular}
	\caption{Probabilistic code-guessing mutants under \(\Dhonref\).  The
		reference scale is \(2^{-L}\) for the tested witness-guess mutant with an
		\(L\)-bit retained code.  The interval column uses Wilson score intervals for
		nonzero counts and the exact Clopper--Pearson upper endpoint for zero counts.}
	\label{tab:probabilistic}
\end{table}

\begin{table}[H]
	\centering
	\small
	\begin{tabular}{S[table-format=2.0] S[table-format=1.0]
			S[table-format=4.0] S[table-format=4.0]
			S[table-format=1.4] l S[table-format=1.4]}
		\toprule
		{Code bits} & {Bytes} & {Trials} & {Passes} & {Est.\ pass prob.} & {95\% interval} & {Ref.\ scale} \\
		\midrule
		1  & 1 & 4096 & 2046 & 0.4995 & {[0.4842, 0.5148]} & 0.5000 \\
		2  & 1 & 4096 & 1045 & 0.2551 & {[0.2420, 0.2687]} & 0.2500 \\
		4  & 1 & 4096 & 273  & 0.0667 & {[0.0594, 0.0747]} & 0.0625 \\
		8  & 1 & 4096 & 21   & 0.0051 & {[0.0034, 0.0078]} & 0.0039 \\
		11 & 2 & 4096 & 3    & 0.0007 & {[0.0002, 0.0022]} & 0.0005 \\
		\bottomrule
	\end{tabular}
	\caption{Selected rows from the August 2025 \hqc{} range-budget sweep.  The full
		artifact sweep covers \(L=1,\ldots,11\) and is plotted in
		\Cref{fig:hqc-sweep}; intervals use the same convention as
		\Cref{tab:probabilistic}.}
	\label{tab:hqc-sweep}
\end{table}

\begin{figure}[H]
	\centering
	\begin{tikzpicture}
		\begin{semilogyaxis}[
			width=0.76\linewidth,
			height=5.2cm,
			xlabel={Retained code bits \(L\)},
			ylabel={Pass probability},
			xmin=1,
			xmax=11,
			ymin=1e-4,
			ymax=1,
			xtick={1,2,3,4,5,6,7,8,9,10,11},
			legend style={font=\scriptsize, at={(0.03,0.04)}, anchor=south west},
			grid=both,
			minor grid style={gray!12},
			major grid style={gray!25},
			]
			\addplot[mark=*, thick, black] coordinates {
				(1,0.499512) (2,0.255127) (3,0.132324) (4,0.066650)
				(5,0.032471) (6,0.015869) (7,0.007568) (8,0.005127)
				(9,0.002930) (10,0.001221) (11,0.0007324)
			};
			\addlegendentry{Observed}
			\addplot[dashed, thick, gray!65!black] coordinates {
				(1,0.5) (2,0.25) (3,0.125) (4,0.0625)
				(5,0.03125) (6,0.015625) (7,0.0078125) (8,0.00390625)
				(9,0.001953125) (10,0.0009765625) (11,0.00048828125)
			};
			\addlegendentry{\(2^{-L}\)}
		\end{semilogyaxis}
	\end{tikzpicture}
	\caption{August 2025 \hqc{} range-budget sweep under \(\Dhonref\).  The observed
		points use the full artifact sweep; selected rows are listed in
		\Cref{tab:hqc-sweep}.  For these witness-guess mutants, retaining more hidden
		code bits reduces the observed pass probability on the order of the \(2^{-L}\)
		reference scale.}
	\label{fig:hqc-sweep}
\end{figure}
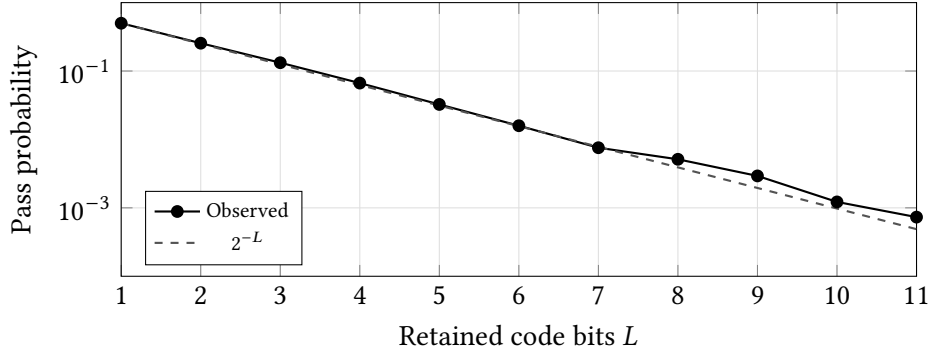

\begin{observation}
	The probabilistic rows are reported against the \(2^{-L}\) reference scale for
	the witness-guess mutants tested here.  The \mlkem{} and \hqc{} one-byte
	variants fall on the corresponding \(2^{-8}\) or \(2^{-5}\) scale, and the
	August 2025 \hqc{} sweep tracks the same order of magnitude through \(L=11\).  The
	HQC-5-\cdstar{} row has zero passes in 4096 trials; the table records the
	finite-run exact upper endpoint, and the \(2^{-37}\) entry records the
	retained range budget for that row.
\end{observation}

\subsection{Mode-selective comparison faults}

\begin{table}[H]
	\centering
	\footnotesize
	\begin{tabularx}{\textwidth}{@{}>{\raggedright\arraybackslash}p{2.8cm}
			>{\raggedright\arraybackslash}p{2.2cm}
			ccc
			>{\raggedright\arraybackslash}X@{}}
		\toprule
		Line & Fault & honest-ref & \(\mathsf{flip}_u\) & \(\mathsf{flip}_v\) & Reading \\
		\midrule
		\mlkem{} plain & compare only \(u\) & undetected (0/3) & undetected (0/3) & detected (3/3) & comparison coverage depends on the corrupted component \\
		\mlkem{} VSFO/\cdone{} & compare only \(u\) & undetected (0/3) & undetected (0/3) & detected (3/3) & the comparison-mode boundary is unchanged \\
		August 2025 \hqc{} plain & compare only \(u\) & undetected (0/3) & undetected (0/3) & detected (3/3) & the August 2025 \hqc{} baseline has the same signature \\
		August 2025 \hqc{} VSFO/\cdone{} & compare only \(u\) & undetected (0/3) & undetected (0/3) & detected (3/3) & code binding leaves this comparison boundary unchanged \\
		\bottomrule
	\end{tabularx}
	\caption{Mode-selective differential signatures for the component-selective
		comparison fault.}
	\label{tab:mode-selective}
\end{table}

\subsection{Summary across campaigns}

\begin{table}[H]
	\centering
	\footnotesize
	\begin{tabularx}{\textwidth}{@{}>{\raggedright\arraybackslash}p{3.0cm}
			>{\raggedright\arraybackslash}p{2.2cm}
			>{\raggedright\arraybackslash}p{2.3cm}
			>{\raggedright\arraybackslash}X@{}}
		\toprule
		Fault family & Honest-reference behavior & Other harness behavior & Interpretation \\
		\midrule
		plain reencryption omission & detected in representative rows & secondary malformed signal & some recomputation faults are already visible before adding a code \\
		transformed binding faults & detected & secondary & the added code creates the honest-reference signal \\
		branch-only decision faults & undetected & detected on malformed or differential support & support-coincidence limitation \\
		component-selective comparison & undetected & detected only on matching modes & comparison coverage is mode-indexed \\
		symmetric endpoint faults & detected against reference & undetected under self-tests & peer-model separation \\
		probabilistic code guesses & quantitative & optional & survival scale follows the retained-code reference scale for the tested mutants \\
		\bottomrule
	\end{tabularx}
	\caption{Family-level summary of the frozen campaigns.}
	\label{tab:family-summary}
\end{table}

At campaign level, the transformed variants add evidence for the code-bearing
good-key path under a named honest-correctness harness and a named support,
with negative rows marking the expected boundaries of that harness.

\section{Scope and Interpretation}
\label{sec:scope}

The positive theorem covers honest-reference tests through localized final-key
list hits, while branch-selection faults that agree with correct decapsulation
on honest ciphertexts are handled by malformed and differential harnesses
within their own supports.

\paragraph{Relative cone closure of the testing layer.}
Within the classical ROM model, honest ciphertext support, black-box
observation boundary, and erasure-containing implementation class specified in
\Cref{sec:dependency-cone-lower-bound}, the honest-reference claim grammar is
closed in the following sense.  Inside the support-active confirmation cone,
positive claims require a bound on the list-hit term, supplied here by cUP,
conditional entropy, hashed-witness entropy, or TailEnt.  Outside that cone,
\Cref{thm:dependency-cone-lower-bound} shows that no black-box honest-support
certifier can prove execution over an implementation class containing the
coupled erasure.  The resulting claim grammar matches the black-box
honest-reference interface under the stated model and erasure-class conditions.

The statements apply to confirmation-code-augmented variants, with standard
\mlkem{} and the August 2025 public \hqc{} line appearing as baselines and
algorithmic anchors.  The augmented variants keep the relevant ciphertext
formats in the artifact, but their good-branch shared secrets include
confirmation-code material.  A standard implementation obtains this testing
signal through an adopted augmented test mode, exported confirmation-code mode,
or corresponding modified line.

\paragraph{ROM versus concrete hash testing.}
The \(2^{-\kappa}\) term in \Cref{thm:code-blind} is a random-oracle-model
statement for a lazy-sampled \(G_K\).  The artifact uses concrete SHA3/SHAKE
implementations with domain separation.  Rows interpreted through the theorem
must satisfy the structural ROM-interface certificate of
\Cref{def:artifact-rom-certificate}, which records the namespace mapping,
serialization, query logging, and absence of hidden witness leakage.  The
experimental rows therefore report concrete pass/fail behavior and empirical
confidence intervals for the tested mutants; they are not empirical
measurements of the ROM fresh-key-coincidence term or proofs that the concrete
primitives are random oracles.

The reproducible mutant campaign runs over concrete \mlkem{} and August 2025
\hqc{} code lines, with timing measurements kept as artifact context, which use the harness outcomes and the probabilistic pass estimates.

Security of the augmented KEM line is cited from the existing
confirmation-code-augmented FO results, and the statements here concern
testing evidence under named harnesses.  For \hqc{}, the paper proves a
retained range budget and states the conditional TailEnt assumption needed to
turn that budget into the entropic theorem.  For \mlkem{}, the source-paper
multi-byte codes carry the cited cUP analysis through the list-cUP bridge; the
artifact's one-byte \cdone{} line is a hashed diagnostic specialization
governed by \Cref{thm:mlkem-cd1-diagnostic}.

The testing theorems are classical ROM statements about harnessed
implementations, with IND-CCA security for \mlkem{} and \hqc{} and QROM
guarantees for the underlying confirmation-code transform cited as
construction-level claims.  Short diagnostic codes are engineering instruments
whose security meaning is captured by the list-hit, min-entropy, and
hashed-witness bounds.  A retained range of \(L\) bits is only a capacity
bound; detectability requires conditional unpredictability of the retained
witness relative to the SUT view.

\section{Conclusion}

Confirmation-code-augmented decapsulation connects correctness testing to the
final key derivation by placing a hidden witness value in the good-branch oracle
input, so a faulty decapsulator that agrees with a sound peer on honest
ciphertexts must compute that witness, query it under the localized prefix, or
pass through the residual freshness term.  In the honest-reference harness,
once correctness errors, aliasing, and freshness defects are separated,
acceptance is bounded by a localized final-key list hit and a fresh-key
coincidence.  The theorem ties the peer model and the sampled support to the
claim, since honest-self tests can mask symmetric endpoint faults, malformed and
differential tests expose branch or component-specific behavior, and repeated
campaigns amplify only under reset, fresh instances, and isolated oracle
namespaces.

The list-hit term is controlled by source cUP, conditional min-entropy, or a
hashed diagnostic witness, according to how the confirmation code is obtained
and what the SUT view contains.  Source cUP applies through a cUP-faithful
interface with no extra witness path; conditional min-entropy yields the
\(q2^{-\lambda}\) bound; hashed diagnostic codes separate raw-witness discovery
from truncated-code guessing; and TailEnt supplies the additional hypothesis
needed to interpret an \hqc{} tail range as a detectability statement.  Under
these hypotheses, the source \mlkem{} confirmation-code line uses the cited cUP
analysis through the cUP-faithful route, the one-byte \mlkem{} artifact line is
a hashed diagnostic experiment, and the August 2025 \hqc{} lines differ by the
number of retained tail bits in \cdone{} and \cdstar{}.

The support-active dependency cone identifies which operations a black-box
honest-reference transcript can certify on a chosen support.  For a black-box
honest-reference harness on honest ciphertext support, operations outside the
confirmation-observable target admit a coupled erasure implementation with the
same transcript distribution over any implementation class containing that
erasure.  Over erasure-closed classes, no execution certifier for such
operations can have soundness and completeness errors summing to less than one.
A negative artifact row has a theorem-covered non-certification interpretation
only when its omitted operation set has a cone-inactivity certificate for the
stated target, support, and class.  Operations outside that certified cone
require a different target, such as a larger retained code, malformed or
differential support, or instrumentation that exposes the missing dependency.

FO-FaultBench records this claim grammar on pinned \mlkem{} and August 2025
\hqc{} code lines.  The stored campaigns show that the derived variants make
the listed missing-code binding faults visible under honest-reference testing,
that branch-only and component-selective faults fall on the supports predicted
by the cone analysis, that symmetric faults separate honest-self from
honest-reference testing, and that probabilistic witness-guess mutants follow
the retained-code scale on the tested distributions.  The resulting testing
theory for confirmation-code-augmented FO KEM variants, paired with the
reproducible implementation package, specifies what a harness proves, which
hypotheses bound the list-hit term, and which operations remain outside
black-box certification until the harness target or support changes.

\printbibliography

@inproceedings{HHK2017,
	author    = {Dennis Hofheinz and Kathrin H{"o}velmanns and Eike Kiltz},
	editor    = {Yael Kalai and Leonid Reyzin},
	title     = {A Modular Analysis of the Fujisaki-Okamoto Transformation},
	booktitle = {Theory of Cryptography -- 15th International Conference, {TCC} 2017, Baltimore, MD, USA, November 12--15, 2017, Proceedings, Part {I}},
	series    = {Lecture Notes in Computer Science},
	volume    = {10677},
	pages     = {341--371},
	publisher = {Springer},
	year      = {2017},
	doi       = {10.1007/978-3-319-70500-2_12},
	url       = {https://doi.org/10.1007/978-3-319-70500-2_12}
}

@inproceedings{BindelEtAl2019TighterQROM,
	author    = {Nina Bindel and Mike Hamburg and Kathrin H{"o}velmanns and Andreas H{"u}lsing and Edoardo Persichetti},
	editor    = {Dennis Hofheinz and Alon Rosen},
	title     = {Tighter Proofs of {CCA} Security in the Quantum Random Oracle Model},
	booktitle = {Theory of Cryptography -- 17th International Conference, {TCC} 2019, Nuremberg, Germany, December 1--5, 2019, Proceedings, Part {II}},
	series    = {Lecture Notes in Computer Science},
	volume    = {11892},
	pages     = {61--90},
	publisher = {Springer},
	year      = {2019},
	doi       = {10.1007/978-3-030-36033-7_3},
	url       = {https://doi.org/10.1007/978-3-030-36033-7_3}
}

@inproceedings{JiangZhangMa2019ExplicitRejection,
	author    = {Haodong Jiang and Zhenfeng Zhang and Zhi Ma},
	editor    = {Dongdai Lin and Kazue Sako},
	title     = {Key Encapsulation Mechanism with Explicit Rejection in the Quantum Random Oracle Model},
	booktitle = {Public-Key Cryptography -- {PKC} 2019 -- 22nd {IACR} International Conference on Practice and Theory of Public-Key Cryptography, Beijing, China, April 14--17, 2019, Proceedings, Part {II}},
	series    = {Lecture Notes in Computer Science},
	volume    = {11443},
	pages     = {618--645},
	publisher = {Springer},
	year      = {2019},
	doi       = {10.1007/978-3-030-17259-6_21},
	url       = {https://doi.org/10.1007/978-3-030-17259-6_21}
}

@inproceedings{DumanEtAl2021PrefixHashing,
	author    = {Julien Duman and Kathrin H{"o}velmanns and Eike Kiltz and Vadim Lyubashevsky and Gregor Seiler},
	title     = {Faster Lattice-Based {KEMs} via a Generic Fujisaki-Okamoto Transform Using Prefix Hashing},
	booktitle = {Proceedings of the 2021 {ACM} {SIGSAC} Conference on Computer and Communications Security},
	pages     = {2722--2737},
	publisher = {Association for Computing Machinery},
	year      = {2021},
	doi       = {10.1145/3460120.3484819},
	url       = {https://doi.org/10.1145/3460120.3484819}
}

@inproceedings{ShanGeXue2022QCCA,
	author    = {Tianshu Shan and Jiangxia Ge and Rui Xue},
	editor    = {Alexandra Boldyreva and Vladimir Kolesnikov},
	title     = {{QCCA}-Secure Generic Transformations in the Quantum Random Oracle Model},
	booktitle = {Public-Key Cryptography -- {PKC} 2023 -- 26th {IACR} International Conference on Practice and Theory of Public-Key Cryptography, Atlanta, GA, USA, May 7--10, 2023, Proceedings, Part {I}},
	series    = {Lecture Notes in Computer Science},
	volume    = {13940},
	pages     = {36--64},
	publisher = {Springer},
	year      = {2023},
	doi       = {10.1007/978-3-031-31368-4_2},
	url       = {https://doi.org/10.1007/978-3-031-31368-4_2}
}

@inproceedings{DingEtAl2022KyberInjectivity,
	author    = {Xiaohui Ding and Muhammed F. Esgin and Amin Sakzad and Ron Steinfeld},
	editor    = {Khoa Nguyen and Guomin Yang and Fuchun Guo and Willy Susilo},
	title     = {An Injectivity Analysis of {Crystals}-{Kyber} and Implications on Quantum Security},
	booktitle = {Information Security and Privacy -- 27th Australasian Conference, {ACISP} 2022, Wollongong, NSW, Australia, November 28--30, 2022, Proceedings},
	series    = {Lecture Notes in Computer Science},
	volume    = {13494},
	pages     = {332--351},
	publisher = {Springer},
	year      = {2022},
	doi       = {10.1007/978-3-031-22301-3_17},
	url       = {https://doi.org/10.1007/978-3-031-22301-3_17}
}

@misc{BarbosaHulsing2023KyberFO,
	author       = {Manuel Barbosa and Andreas H{"u}lsing},
	title        = {The security of Kyber's {FO}-transform},
	howpublished = {Cryptology {ePrint} Archive, Paper 2023/755},
	year         = {2023},
	url          = {https://eprint.iacr.org/2023/755}
}

@inproceedings{GeShanXue2023ExplicitRejection,
	author    = {Jiangxia Ge and Tianshu Shan and Rui Xue},
	editor    = {Helena Handschuh and Anna Lysyanskaya},
	title     = {Tighter {QCCA}-Secure Key Encapsulation Mechanism with Explicit Rejection in the Quantum Random Oracle Model},
	booktitle = {Advances in Cryptology -- {CRYPTO} 2023 -- 43rd Annual International Cryptology Conference, {CRYPTO} 2023, Santa Barbara, CA, USA, August 20--24, 2023, Proceedings, Part {V}},
	series    = {Lecture Notes in Computer Science},
	volume    = {14085},
	pages     = {292--324},
	publisher = {Springer},
	year      = {2023},
	doi       = {10.1007/978-3-031-38554-4_10},
	url       = {https://doi.org/10.1007/978-3-031-38554-4_10}
}

@inproceedings{GeLiaoXue2024MRE,
	author    = {Jiangxia Ge and Heming Liao and Rui Xue},
	editor    = {Kai-Min Chung and Yu Sasaki},
	title     = {Measure-Rewind-Extract: Tighter Proofs of One-Way to Hiding and {CCA} Security in the Quantum Random Oracle Model},
	booktitle = {Advances in Cryptology -- {ASIACRYPT} 2024 -- 30th International Conference on the Theory and Application of Cryptology and Information Security, Kolkata, India, December 9--13, 2024, Proceedings, Part {IV}},
	series    = {Lecture Notes in Computer Science},
	volume    = {15487},
	pages     = {3--34},
	publisher = {Springer},
	year      = {2025},
	doi       = {10.1007/978-981-96-0894-2_1},
	url       = {https://doi.org/10.1007/978-981-96-0894-2_1}
}

@inproceedings{FischlinGuenther2023VV,
	author    = {Marc Fischlin and Felix G{"u}nther},
	title     = {Verifiable Verification in Cryptographic Protocols},
	booktitle = {Proceedings of the 2023 {ACM} {SIGSAC} Conference on Computer and Communications Security},
	pages     = {3239--3253},
	publisher = {Association for Computing Machinery},
	year      = {2023},
	doi       = {10.1145/3576915.3623151},
	url       = {https://doi.org/10.1145/3576915.3623151}
}

@inproceedings{Glabush2025VerifiableDecapsulation,
	author    = {Lewis Glabush and Felix G{"u}nther and Kathrin H{"o}velmanns and Douglas Stebila},
	editor    = {Yael Tauman Kalai and Seny F. Kamara},
	title     = {Verifiable Decapsulation: Recognizing Faulty Implementations of Post-quantum {KEMs}},
	booktitle = {Advances in Cryptology -- {CRYPTO} 2025 -- 45th Annual International Cryptology Conference, Santa Barbara, CA, USA, August 17--21, 2025, Proceedings, Part {III}},
	series    = {Lecture Notes in Computer Science},
	volume    = {16002},
	pages     = {543--574},
	publisher = {Springer},
	year      = {2025},
	doi       = {10.1007/978-3-032-01881-6_17},
	url       = {https://doi.org/10.1007/978-3-032-01881-6_17}
}

@misc{GlabushHoevelmannsStebila2025MultiTarget,
	author       = {Lewis Glabush and Kathrin H{"o}velmanns and Douglas Stebila},
	title        = {Tight Multi-challenge Security Reductions for Key Encapsulation Mechanisms},
	howpublished = {Cryptology {ePrint} Archive, Paper 2025/343},
	year         = {2025},
	url          = {https://eprint.iacr.org/2025/343}
}

@inproceedings{HovelmannsHulsingMajenz2022FailingGracefully,
	author    = {Kathrin H{"o}velmanns and Andreas H{"u}lsing and Christian Majenz},
	editor    = {Shweta Agrawal and Dongdai Lin},
	title     = {Failing Gracefully: Decryption Failures and the Fujisaki-Okamoto Transform},
	booktitle = {Advances in Cryptology -- {ASIACRYPT} 2022 -- 28th International Conference on the Theory and Application of Cryptology and Information Security, Taipei, Taiwan, December 5--9, 2022, Proceedings, Part {IV}},
	series    = {Lecture Notes in Computer Science},
	volume    = {13794},
	pages     = {414--443},
	publisher = {Springer},
	year      = {2022},
	doi       = {10.1007/978-3-031-22972-5_15},
	url       = {https://doi.org/10.1007/978-3-031-22972-5_15}
}

@inproceedings{HovelmannsMajenz2023FailingGracefullyNote,
	author    = {Kathrin H{"o}velmanns and Christian Majenz},
	editor    = {Markku-Juhani Saarinen and Daniel Smith-Tone},
	title     = {A Note on Failing Gracefully: Completing the Picture for Explicitly Rejecting Fujisaki-Okamoto Transforms Using Worst-Case Correctness},
	booktitle = {Post-Quantum Cryptography -- 15th International Conference, {PQCrypto} 2024, Oxford, UK, June 12--14, 2024, Proceedings, Part {II}},
	series    = {Lecture Notes in Computer Science},
	volume    = {14772},
	pages     = {245--265},
	publisher = {Springer},
	year      = {2024},
	doi       = {10.1007/978-3-031-62746-0_11},
	url       = {https://doi.org/10.1007/978-3-031-62746-0_11}
}

@inproceedings{HovelmannsKudinov2025ExplicitImplicit,
	author    = {Kathrin H{"o}velmanns and Mikhail A. Kudinov},
	editor    = {Ruben Niederhagen and Markku-Juhani O. Saarinen},
	title     = {Treating Dishonest Ciphertexts in Post-quantum {KEMs} -- Explicit vs. Implicit Rejection in the {FO} Transform},
	booktitle = {Post-Quantum Cryptography -- 16th International Workshop, {PQCrypto} 2025, Taipei, Taiwan, April 8--10, 2025, Proceedings, Part {II}},
	series    = {Lecture Notes in Computer Science},
	volume    = {15578},
	pages     = {325--350},
	publisher = {Springer},
	year      = {2025},
	doi       = {10.1007/978-3-031-86602-9_12},
	url       = {https://doi.org/10.1007/978-3-031-86602-9_12}
}

@inproceedings{DonEtAl2021OnlineExtractability,
	author    = {Jelle Don and Serge Fehr and Christian Majenz and Christian Schaffner},
	editor    = {Orr Dunkelman and Stefan Dziembowski},
	title     = {Online-Extractability in the Quantum Random-Oracle Model},
	booktitle = {Advances in Cryptology -- {EUROCRYPT} 2022 -- 41st Annual International Conference on the Theory and Applications of Cryptographic Techniques, Trondheim, Norway, May 30--June 3, 2022, Proceedings, Part {III}},
	series    = {Lecture Notes in Computer Science},
	volume    = {13277},
	pages     = {677--706},
	publisher = {Springer},
	year      = {2022},
	doi       = {10.1007/978-3-031-07082-2_24},
	url       = {https://doi.org/10.1007/978-3-031-07082-2_24}
}

@misc{GeShanXue2023FOQCCALift,
	author       = {Jiangxia Ge and Tianshu Shan and Rui Xue},
	title        = {On the Fujisaki-Okamoto transform: from Classical {CCA} Security to Quantum {CCA} Security},
	howpublished = {Cryptology {ePrint} Archive, Paper 2023/792},
	year         = {2023},
	url          = {https://eprint.iacr.org/2023/792}
}

@inproceedings{MajenzSisinni2024FFPNG,
	author    = {Christian Majenz and Fabrizio Sisinni},
	editor    = {Leonid Reyzin and Douglas Stebila},
	title     = {Provable Security Against Decryption Failure Attacks from {LWE}},
	booktitle = {Advances in Cryptology -- {CRYPTO} 2024 -- 44th Annual International Cryptology Conference, Santa Barbara, CA, USA, August 18--22, 2024, Proceedings, Part {II}},
	series    = {Lecture Notes in Computer Science},
	volume    = {14921},
	pages     = {456--485},
	publisher = {Springer},
	year      = {2024},
	doi       = {10.1007/978-3-031-68379-4_14},
	url       = {https://doi.org/10.1007/978-3-031-68379-4_14}
}

@inproceedings{HovelmannsHulsingMajenzSisinni2025ReEncryptionAlternatives,
	author    = {Kathrin H{"o}velmanns and Andreas H{"u}lsing and Christian Majenz and Fabrizio Sisinni},
	editor    = {Serge Fehr and Pierre-Alain Fouque},
	title     = {{(Un)breakable Curses} -- Re-encryption in the Fujisaki-Okamoto Transform},
	booktitle = {Advances in Cryptology -- {EUROCRYPT} 2025 -- 44th Annual International Conference on the Theory and Applications of Cryptographic Techniques, Madrid, Spain, May 4--8, 2025, Proceedings, Part {II}},
	series    = {Lecture Notes in Computer Science},
	volume    = {15602},
	pages     = {245--274},
	publisher = {Springer},
	year      = {2025},
	doi       = {10.1007/978-3-031-91124-8_9},
	url       = {https://doi.org/10.1007/978-3-031-91124-8_9}
}

@techreport{NISTSP800227,
	author      = {Gorjan Alagic and Elaine Barker and Lily Chen and Dustin Moody and Angela Robinson and Hamilton Silberg and Noah Waller},
	title       = {Recommendations for Key-Encapsulation Mechanisms},
	institution = {National Institute of Standards and Technology},
	type        = {{NIST} Special Publication},
	number      = {{NIST SP} 800-227},
	year        = {2025},
	month       = sep,
	day         = {18},
	doi         = {10.6028/NIST.SP.800-227},
	url         = {https://csrc.nist.gov/pubs/sp/800/227/final}
}

@techreport{NISTFIPS203,
	author      = {{National Institute of Standards and Technology}},
	title       = {Module-Lattice-Based Key-Encapsulation Mechanism Standard},
	institution = {National Institute of Standards and Technology},
	type        = {Federal Information Processing Standards Publication},
	number      = {{NIST FIPS} 203},
	year        = {2024},
	month       = aug,
	day         = {13},
	doi         = {10.6028/NIST.FIPS.203},
	url         = {https://csrc.nist.gov/pubs/fips/203/final}
}

@misc{HQC2025Spec,
	author       = {Philippe Gaborit and Carlos Aguilar-Melchor and Nicolas Aragon and Slim Bettaieb and Lo{"i}c Bidoux and Olivier Blazy and Jean-Christophe Deneuville and Edoardo Persichetti and Gilles Z{'e}mor and Jurjen Bos and Arnaud Dion and J{'e}r{^o}me Lacan and Jean-Marc Robert and Pascal V{'e}ron and Paulo L. Barreto and Santosh Ghosh and Shay Gueron and Tim G{"u}neysu and Rafael Misoczki and Jan Richter-Brokmann and Nicolas Sendrier and Jean-Pierre Tillich and Valentin Vasseur},
	title        = {Hamming Quasi-Cyclic ({HQC})},
	howpublished = {{HQC} specifications},
	year         = {2025},
	month        = aug,
	day          = {22},
	url          = {https://pqc-hqc.org/doc/hqc_specifications_2025_08_22.pdf}
}

\appendix

\section{Additional HQC Sweep}
\label{app:hqc-sweep}

\begin{table}[htbp]
	\centering
	\scriptsize
	\begin{adjustbox}{max width=\linewidth}
		\begin{tabular}{S[table-format=2.0] S[table-format=1.0]
				S[table-format=4.0] S[table-format=4.0]
				S[table-format=1.4] l S[table-format=1.4]}
			\toprule
			{Bits} & {Bytes} & {Trials} & {Passes} & {Est.\ prob.} & {95\% CI / ub} & {Ref.\ scale} \\
			\midrule
			1 & 1 & 4096 & 2046 & 0.4995 & {[0.4842, 0.5148]} & 0.5000 \\
			2 & 1 & 4096 & 1045 & 0.2551 & {[0.2420, 0.2687]} & 0.2500 \\
			3 & 1 & 4096 & 542 & 0.1323 & {[0.1223, 0.1430]} & 0.1250 \\
			4 & 1 & 4096 & 273 & 0.0667 & {[0.0594, 0.0747]} & 0.0625 \\
			5 & 1 & 4096 & 133 & 0.0325 & {[0.0275, 0.0384]} & 0.0313 \\
			6 & 1 & 4096 & 65 & 0.0159 & {[0.0125, 0.0202]} & 0.0156 \\
			7 & 1 & 4096 & 31 & 0.0076 & {[0.0053, 0.0107]} & 0.0078 \\
			8 & 1 & 4096 & 21 & 0.0051 & {[0.0034, 0.0078]} & 0.0039 \\
			9 & 2 & 4096 & 12 & 0.0029 & {[0.0017, 0.0051]} & 0.0020 \\
			10 & 2 & 4096 & 5 & 0.0012 & {[0.0005, 0.0029]} & 0.0010 \\
			11 & 2 & 4096 & 3 & 0.0007 & {[0.0002, 0.0022]} & 0.0005 \\
			\bottomrule
		\end{tabular}
	\end{adjustbox}
	\caption{Full August 2025 \hqc{} range-budget sweep for \(L=1,\ldots,11\), with
		the same interval convention as the main probabilistic tables.}
	\label{tab:hqc-sweep-full}
\end{table}

\end{document}